\documentclass[11pt,a4paper]{article}

\usepackage[utf8]{inputenc}

\usepackage[dvipsnames]{xcolor}

\usepackage{amsmath}
\usepackage{amsfonts}
\usepackage{amssymb}
\usepackage{amsthm}
\usepackage{enumitem}

\newtheorem{thm}{Theorem}[section]
\newtheorem{lem}{Lemma}[section]
\newtheorem{sublemma-lem}{}[lem]
\newtheorem{sublemma}{}[thm]
\newtheorem{proposition}{Proposition}[section]
\newtheorem{sublemma-prop}{}[proposition]
\newtheorem{obs}{Observation}[section]
\newtheorem{cor}{Corollary}[section]

\usepackage{multicol}

\title{Placing quantified variants of \textsc{3-SAT} and \textsc{Not-All-Equal 3-SAT} in the polynomial hierarchy}

\author{Janosch D\"{o}cker, Britta Dorn, Simone Linz, and Charles Semple}

\date{\today}

\begin{document}

\maketitle

\begin{abstract}
The complexity of variants of 3-SAT and \textsc{Not-All-Equal 3-SAT} is well studied. However, in contrast, very little is known about the complexity of the problems' quantified counterparts. In the first part of this paper, we show that  \textsc{$\forall \exists$~3-SAT}  is $\Pi_2^P$-complete even if (1) each variable appears exactly twice unnegated and exactly twice negated, (2) each clause is a disjunction of exactly three distinct variables, and (3) the number of universal variables is equal to the number of existential variables. Furthermore, we show that the problem remains $\Pi_2^P$-complete if (1a) each universal variable appears exactly once unnegated and exactly once negated, (1b) each existential variable appears exactly twice unnegated and exactly twice negated, and (2) and (3) remain unchanged. On the other hand, the problem becomes NP-complete for certain variants in which each universal variable appears exactly once. In the second part of the paper, we establish $\Pi_2^P$-completeness for  \textsc{$\forall \exists$ Not-All-Equal 3-SAT} even if (1') the Boolean formula is linear and monotone, (2') each universal variable appears exactly once and each existential variable appears exactly three times, and (3') each clause is a disjunction of exactly three distinct variables that contains at most one universal variable. On the positive side, we uncover variants of \textsc{$\forall \exists$ Not-All-Equal 3-SAT} that are co-NP-complete or solvable in polynomial time.
\end{abstract}


\section{Introduction}
The Boolean satisfiability problem \textsc{SAT} plays a major role in the theory of NP-completeness. It was the first problem shown to be complete for the class NP (Cook's Theorem~\cite{cook71}) and many NP-hardness results are established by using this problem, or restricted variants thereof, as a starting point for polynomial-time reductions. Restricted variants of a problem that remain NP-complete are particularly useful as they allow for the possibility of simpler proofs and stronger results.

The most prominent NP-complete variant of the Boolean satisfiability problem is \textsc{3-SAT}. Here we are given a conjunction of clauses  such that each clause contains exactly three literals, where a literal is a propositional variable or its negation. An instance of \textsc{3-SAT} is a yes-instance if there is a truth assignment to the propositional variables\footnote{From now on, we simply say variable instead of propositional variable since all variables used in the paper take only values representing true and false.} such that at least one literal of each clause evaluates to true. Interestingly, even within {\sc 3-SAT}, we can restrict the problem further. For example, for instances of 3-SAT in which each clause contains exactly three distinct variables, Tovey~\cite[Theorem 2.3]{tovey84} proved that \textsc{3-SAT} remains NP-complete if each variable appears in at most four clauses. Furthermore, this result also holds if each variable appears exactly twice unnegated and exactly twice negated~\cite[Theorem 1]{berman03}. On the other hand, the problem becomes trivial, i.e., the answer is always yes, if each variable appears at most three times~\cite[Theorem 2.4]{tovey84}.

A popular NP-complete variant of \textsc{3-SAT} called \textsc{Not-All-Equal 3-SAT} (\textsc{NAE-3-SAT}) asks whether we can assign truth values to the variables such that at least one, but not all, of the literals of each clause evaluate to true. Schaefer~\cite{schaefer78} first established NP-completeness of NAE-3-SAT in the setting where each clause contains at most three literals. Subsequently, Karpinski and Piecuch~\cite{karpinski17,karpinski18} showed that \textsc{NAE-3-SAT} is NP-complete if each variable appears at most four times. Furthermore, Porschen et al.~\cite[Theorem 3]{porschen14} showed that \textsc{NAE-3-SAT} remains NP-complete if (i) each literal appears at most once in any clause, and (ii) the input formula is {\it linear} and {\it monotone}, that is, each pair of distinct clauses share at most one variable and no clause contains a literal that is the negation of some variable. Following on from this last result, Darmann and D\"ocker~\cite{darmann19} showed recently that \textsc{NAE-3-SAT} remains NP-complete if, in addition to (i) and (ii), each variable appears exactly four times. By contrast, if a monotone conjunction of clauses has the property that each variable appears at most  three times, \textsc{NAE-3-SAT} can be decided in linear time~\cite[Theorem 4, p.\,186]{porschen04}. 

In this paper, we consider generalized variants of \textsc{3-SAT} and \textsc{NAE-3-SAT}, namely \textsc{$\forall \exists$ 3-SAT} and \textsc{$\forall \exists$ NAE-3-SAT}, respectively. Briefly, {\sc $\forall \exists$ 3-SAT} is a quantified variant of {\sc 3-SAT}, where each variable is either {\it universal} or {\it existential}. The decision problem \textsc{$\forall \exists$ 3-SAT} asks if, for every assignment of truth values to the universal variables, there exists an assignment of truth values to the existential variables such that at least one literal of each clause evaluates to true. Observe that, if an instance of \textsc{$\forall \exists$ 3-SAT} does not contain a universal variable, then this instance reduces to an instance of {\sc 3-SAT}. Analogously, we can think of {\sc $\forall \exists$ NAE-3-SAT} as a generalized variant of \textsc{NAE-3-SAT}. Formal definitions of both problems are given in the next section.

Stockmeyer~\cite{stockmeyer76} and Dahlhaus et al.~\cite{eiter95} showed, respectively, that {\sc $\forall \exists$ 3-SAT} and {\sc $\forall \exists$ NAE-3-SAT} are complete for the second level of the polynomial hierarchy or, more precisely, complete for the class $\Pi_2^P$. In this paper, we establish $\Pi_2^P$-completeness for restricted variants of these two quantified problems. In particular, we show that \textsc{$\forall \exists$ 3-SAT} is $\Pi_2^P$-complete if each universal variable appears exactly once unnegated and exactly once negated, and each existential variable appears exactly twice unnegated and exactly once negated or each existential variable appears exactly once unnegated and exactly twice negated. Although we do not consider approximation aspects in this paper, by way of comparison, Haviv et al.~\cite{haviv07} showed that approximating a particular optimization version of \textsc{$\forall \exists$ 3-SAT} is $\Pi_2^P$-hard even if each universal variable appears at most twice and each existential variable appears at most three times. Whether optimization versions of the $\Pi_2^P$-complete problems presented in this paper are $\Pi_2^P$-hard to approximate is a question that we leave for future research. 
Furthermore, we establish $\Pi_2^P$-completeness for \textsc{$\forall \exists$ 3-SAT} if each universal variable appears exactly $s_1$ times unnegated and exactly $s_2$ times negated, each existential variable appears exactly $t_1$ times unnegated and exactly $t_2$ times negated, and the following three properties are satisfied: (i) $s_1=s_2$, (ii) $s_1\in\{1,2\}$, and (iii) $t_1=t_2=2$. These latter completeness results hold even if each clause is a disjunction of exactly three distinct variables and the number of universal and existential variables is {\it balanced}, that is, the number of universal and existential variables are the same.

Turning to \textsc{$\forall \exists$ NAE-3-SAT}, we show that the problem remains $\Pi_2^P$-complete if each universal variable appears exactly once, each clause contains at most one universal variable, each existential variable appears exactly three times, and the conjunction of clauses is both linear and monotone. Interestingly, while one appearance of each universal variable is enough to obtain a $\Pi_2^P$-hardness result in this setting, the same is not true for \textsc{$\forall \exists$ 3-SAT} unless the polynomial hierarchy collapses~\cite[p.\,55]{haviv07}.

The remainder of the paper is organized as follows. The next section introduces notation and terminology, and formally states  three variants of {\sc $\forall \exists$ 3-SAT} and {\sc $\forall \exists$ Not-All-Equal 3-SAT} that are the main focus of this paper. Section 3 (resp. Section 4) investigates the computational complexity of {\sc $\forall \exists$ 3-SAT} (resp. {\sc $\forall \exists$ Not-All-Equal 3-SAT).} Both of Sections 3 and 4 start with a subsection on enforcers that are needed for the subsequent hardness proofs and, in terms of future work, we expect to be of independent interest in their own right.

\section{Preliminaries} \label{sec:prelim}
This section introduces notation and terminology that is used throughout the paper.

Let $V = \{x_1, x_2, \ldots, x_n\}$ be a set of variables. A \emph{literal} is a variable or its negation. We denote the set  $\{x_i, \overline{x}_i: i\in \{1, 2, \ldots, n\}\}$ of all literals that correspond to elements in $V$ by $\mathcal{L}_V$. A \emph{clause} is a disjunction of a subset of $\mathcal{L}_V$. If a clause contains exactly $k$ distinct literals for $k\ge 1$, then it is called a {\em $k$-clause}. For example, $(x_1 \vee \bar{x}_2 \vee x_4)$ is a $3$-clause. A {\it Boolean formula in conjunctive normal form} (CNF) is a conjunction of clauses, i.e., an expression of the form $\varphi = \bigwedge_{j = 1}^m C_j$, where $C_j$ is a clause for all $j$. In what follows, we refer to a Boolean formula in conjunctive normal form simply as a {\it Boolean formula}. Now, let $\varphi$ be a Boolean formula. We say that $\varphi$ is \emph{linear} if any pair of distinct clauses share at most one variable and that it is {\em monotone} if no clause contains an element in $\{\overline{x}_1,\overline{x}_2,\ldots,\overline{x}_n\}$. Furthermore, if each clause contains at most $k$ literals, it is said to be in {\em $k$-CNF}. For each variable $x_i \in V$, we denote the number of appearances of $x_i$ in~$\varphi$ plus the number of appearances of $\overline{x}_i$ in~$\varphi$ by $a(x_i)$. A {\em variable assignment} or, equivalently, {\em truth assignment} for $V$ is a mapping $\beta \colon V \rightarrow \{T, F\}$, where~$T$ represents the truth value True and $F$ represents the truth value False. A truth assignment $\beta$ \emph{satisfies} $\varphi$ if at least one literal of each clause evaluates to $T$ under $\beta$. If there exists a truth assignment that satisfies~$\varphi$, we say that $\varphi$ is {\it satisfiable}. For a truth assignment $\beta$ that satisfies $\varphi$ and has the additional property that at least one literal of each clause evaluates to $F$, we say that $\beta$ \emph{nae-satisfies} $\varphi$. Lastly, let $V$ and $V'$ be two disjoint sets of variables, let $\beta$ be a truth assignment for $V$, and let $\beta'$ be a truth assignment for $V\cup V'$. We say that $\beta'$ {\it extends} $\beta$ (or, alternatively, that $\beta$ extends to $V\cup V'$) if $\beta(x_i)=\beta'(x_i)$ for each $x_i\in V$.

A {\em quantified Boolean formula $\Phi$} over a set $V=\{x_1, x_2, \ldots, x_n\}$ of variables is a formula of the form
$$\forall x_1 \cdots \forall x_p \exists x_{p+1} \cdots \exists x_n \bigwedge_{j=1}^m C_j.$$  
The variables $x_1, x_2, \ldots, x_p$ are  \emph{universal} variables of $\Phi$ and the variables $x_{p+1}, x_{p+2}, \ldots, x_n$ are  \emph{existential} variables of $\Phi$. Furthermore, for variables $x_{i}, x_{i+1}, \ldots, x_{i'}$ with $1 \leq i < i' \leq p$ and $x_{i''}, x_{i''+1}, \ldots, x_{i'''}$ with $p+1 \leq i'' < i''' \leq n$, we define
 \[
\forall X_i^{i'} := \forall x_i \cdots \forall x_{i'}  \text{\hspace{3mm}and\hspace{3mm}} \exists X_{i''}^{i'''} := \exists x_{i''} \cdots \exists x_{i'''},
\]
respectively and, similarly,
 \[
X_i^{i'} := \{x_i, \ldots,  x_{i'}\}  \text{\hspace{3mm}and\hspace{3mm}} X_{i''}^{i'''} := \{x_{i''} , \ldots, x_{i'''}\},
\]
respectively.

We next introduce notation that transforms a Boolean formula $\varphi$ into another such formula. Specifically, we use $\varphi[x \mapsto y]$ to denote the Boolean formula obtained from $\varphi$ by replacing each appearance of variable $x$ with variable $y$ (i.e., replace $x$ with $y$ and replace $\overline{x}$ with $\overline{y}$). For pairwise distinct pairs $(x_1, y_1), (x_2, y_2), \ldots, (x_k, y_k)$ of variables, we use $\varphi[x_1 \mapsto y_1, \ldots, x_k \mapsto y_k]$ to denote the Boolean formula obtained from $\varphi$ by simultaneously replacing each appearance of variable $x_i$ by variable $y_i$ for $1 \leq i \leq k$. Since the variables are pairwise distinct, note that this operation is well-defined.  Lastly, for a constant $b \in \{T, F\}$, the Boolean formula $\varphi[x \mapsto b]$ is obtained from $\varphi$ by replacing each appearance of variable $x$ by $b$. \\
  
\noindent {\bf The polynomial hierarchy.} An {\em oracle} for a complexity class $A$ is a black box that, in constant time, outputs the answer to any given instance of a decision problem contained in $A$. The {\em polynomial hierarchy} is a system of nested complexity classes that are defined recursively as follows. Set
\[\Sigma_0^P = \Pi_0^P = \text{P},\]
and define, for all $k\ge 0$,
\[\Sigma_{k+1}^P = \text{NP}^{\Sigma_k^P}\text{\hspace{3mm}and\hspace{3mm}}\Pi_{k+1}^P = \text{co-NP}^{\Sigma_k^P},\]
where a problem is in $\text{NP}^{\Sigma_k^P}$ (resp.\ $\text{co-NP}^{\Sigma_k^P}$) if we can verify an appropriate certificate of a yes-instance (resp.\ no-instance) in polynomial time when given access to an oracle for $\Sigma_k^P$. By definition, $\Sigma_1^P = \text{NP}$ and $\Pi_1^P = \text{co-NP}$. We say that the classes $\Sigma_k^P$ and $\Pi_k^P$ are on the {\em $k$-th level} of the polynomial hierarchy.

For all $k \geq 0$, there are complete problems under polynomial-time many-one reductions for $\Sigma_k^P$ and $\Pi_k^P$. However, while, for example, the complexity class $\Pi_2^P$ generalizes both NP and co-NP, it remains an open question whether $\Sigma_k^P \neq \Sigma_{k+1}^P$ or $\Pi_k^P \neq \Pi_{k+1}^P$ for any $k \geq 0$. For further details of the polynomial hierarchy, we refer the interested reader to Garey and Johnson's book~\cite{garey79}, an article by Stockmeyer~\cite{stockmeyer76}, as well as to the compendium by Schaefer and Umans~\cite{schaefer02} for a collection of problems that are complete  for the second or higher levels of the polynomial hierarchy.\\

The following two $\Pi_2^P$-complete problems  are the starting points for the work presented in this paper.

\noindent {\sc $\forall \exists$ 3-SAT}\\
\noindent{\bf Input.} A quantified Boolean formula
$$\forall x_1 \cdots \forall x_p \exists x_{p+1} \cdots \exists x_n \bigwedge_{j=1}^m C_j$$
over a set  $V=\{x_1,x_2,\ldots,x_n\}$ of variables, where each clause $C_j$ is a disjunction of at most three literals.
\\
\noindent{\bf Question.} For every truth assignment for $\{x_1, x_2, \ldots, x_p\}$, does there exist a truth assignment for $\{x_{p+1}, x_{p+2}, \ldots, x_n\}$ such that each clause of the formula is satisfied?

\noindent {\sc $\forall \exists$ Not-All-Equal 3-SAT} ({\sc $\forall \exists$ NAE-3-SAT})\\
\noindent{\bf Input.} A quantified Boolean formula
$$\forall x_1 \cdots \forall x_p \exists x_{p+1} \cdots \exists x_n \bigwedge_{j=1}^m C_j$$
over a set  $V=\{x_1,x_2,\ldots,x_n\}$ of variables, where each clause $C_j$ is a disjunction of at most three literals.
\\
\noindent{\bf Question.} For every truth assignment for $\{x_1, x_2, \ldots, x_p\}$, does there exist a truth assignment for $\{x_{p+1}, x_{p+2}, \ldots, x_n\}$ such that each clause of the formula is nae-satisfied?

\noindent Stockmeyer~\cite{stockmeyer76}, and Eiter and Gottlob~\cite{eiter95} established $\Pi_2^P$-completeness for {\sc $\forall \exists$ 3-SAT} and {\sc $\forall \exists$ NAE-3-SAT}, respectively.

The main focus of this paper are the following three restricted variants of {\sc $\forall \exists$ 3-SAT} and {\sc $\forall \exists$ NAE-3-SAT}.

\noindent {\sc Balanced $\forall \exists$ 3-SAT-$(s_1,s_2, t_1, t_2)$}\\
\noindent{\bf Input.} Four non-negative integers $s_1, s_2, t_1,t_2$, and a quantified Boolean formula
$$\forall x_1 \cdots \forall x_p \exists x_{p+1} \cdots \exists x_n \bigwedge_{j=1}^m C_j$$
over a set  $V=\{x_1,x_2,\ldots,x_n\}$ of variables such that (i) $n = 2p$, (ii) each $C_j$ is a 3-clause that contains three \emph{distinct} variables, and (iii), amongst the clauses, each universal variable appears unnegated exactly $s_1$ times and negated exactly $s_2$ times, and each existential variable appears unnegated exactly $t_1$ times and negated exactly $t_2$ times.
\\
\noindent{\bf Question.} For every truth assignment for $\{x_1, x_2, \ldots, x_p\}$, does there exist a truth assignment for $\{x_{p+1}, x_{p+2}, \ldots, x_n\}$ such that each clause of the formula is satisfied?

\noindent In addition, we also consider the decision problem that is obtained from {\sc Balanced $\forall \exists$ 3-SAT-$(s_1,s_2, t_1, t_2)$} by omitting property (i) in the statement of the input. We refer to the resulting problem as {\sc $\forall \exists$ 3-SAT-$(s_1,s_2, t_1, t_2)$}. Lastly, we consider the following problem.

\noindent {\sc Monotone $\forall \exists$ NAE-3-SAT-$(s, t)$}\\
\noindent{\bf Input.} Two non-negative integers $s$ and $t$, and a monotone quantified Boolean formula
$$\forall x_1 \cdots \forall x_p \exists x_{p+1} \cdots \exists x_n \bigwedge_{j=1}^m C_j$$
over a set  $V=\{x_1,x_2,\ldots,x_n\}$ of variables such that (i) each $C_j$ is a 3-clause that contains  three \emph{distinct} variables and (ii), amongst the clauses, each universal variable appears exactly $s$ times and each existential variable appears exactly $t$ times. \\
\noindent{\bf Question.} For every truth assignment for $\{x_1, x_2, \ldots, x_p\}$, does there exist a truth assignment for $\{x_{p+1}, x_{p+2}, \ldots, x_n\}$ such that each clause of the formula is nae-satisfied?

\noindent {\bf Enforcers.} To establish the results of this paper, we will frequently use the concept of enforcers. An \emph{enforcer} (sometimes also called a {\it gadget})~\cite{berman03} is a Boolean formula, where the formula itself and each truth assignment that satisfies it have a certain structure. Enforcers are used in  polynomial-time reductions to add additional restrictions on how yes-instances can be obtained.

We next detail two unquantified enforcers that were introduced by Ber\-man et al.~\cite[p.\,3]{berman03} and that lay the foundation for several other enforcers that are new to this paper and will be introduced in the following sections. First, let $\ell_1,\ell_2$ and $\ell_3$ be three, not necessarily distinct, literals. Without loss of generality, we may assume that $\ell_1\in\{x_1,\overline{x}_1\}$, $\ell_2\in\{x_2,\overline{x}_2\}$, and $\ell_3\in\{x_3,\overline{x}_3\}$. Now consider the following enforcer to which we  refer to as  $\mathcal{S}$-{\it enforcer}:
\begin{align*}
\mathcal{S}(\ell_1, \ell_2, \ell_3) = &(\ell_1 \vee \overline{a} \vee b) \wedge (\ell_2 \vee \overline{b} \vee c) \wedge (\ell_3 \vee a \vee \overline{c}) \wedge {} \\ 
&(a \vee b \vee c) \wedge (\overline{a} \vee \overline{b} \vee \overline{c}), 
\end{align*}
where $a$, $b$, and $c$ are new variables such that $\{x_1,x_2,x_3\}\cap\{a,b,c\}=\emptyset$. Let $\beta \colon \{x_1,x_2,x_3,a,b,c\} \rightarrow \{T, F\}$ be a truth assignment. The next observation is an immediate consequence from the fact that, if $\beta(\ell_1)=\beta(\ell_2)=\beta(\ell_3)=F$, then, as the first three clauses form a cyclic implication chain which can only be satisfied by setting $\beta(a)=\beta(b)=\beta(c)$, either the fourth or fifth clause is not satisfied.

\begin{obs}\label{ob:enforcer1}
Consider the boolean formula $\mathcal{S}(\ell_1, \ell_2, \ell_3)$, where $\ell_i \in \{x_i, \overline{x_i}\}$, and let $V$ be its associated set of variables. A truth assignment $\beta$ for the variables $\{x_1, x_2, x_3\}$ can be extended to a truth assignment $\beta'$ for $V$ that satisfies $\mathcal{S}(\ell_1, \ell_2, \ell_3)$ if and only if $\beta(\ell_i)=T$ for some $i\in\{1,2,3\}$.
\end{obs}
\noindent {\bf Remark.} We note that Observation~\ref{ob:enforcer1} holds, even if $x_1$ is a universal variable and all other variables in $\{x_2,x_3,a,b,c\}$ are existential in which case we will write $\mathcal{S}_u(\ell_1, \ell_2, \ell_3)$ to denote the gadget.

In what follows, we will use enforcers that are built of several copies of the $\mathcal{S}$-enforcer. In such a case, for each pair of $\mathcal{S}$-enforcer copies, the two 3-element sets of new variables are disjoint. 

Again following the constructions from Berman et al.~\cite{berman03}, consider a second enforcer:
\[
x^{(2)} = \mathcal{S}(x, y, y) \wedge \mathcal{S}(x, \bar{y}, \bar{y}).
\]
Note that $x^{(2)}$ is a Boolean formula over eight variables. Moreover, each clause contains three distinct variables since the copies of $y$ and $\bar{y}$ are distributed over different clauses in $\mathcal{S}(x, y, y)$ and $\mathcal{S}(x, \bar{y}, \bar{y})$, respectively. Lastly, each variable, except for $x$, appears exactly twice unnegated and twice negated in $x^{(2)}$. Now, the next observation follows by construction and Observation \ref{ob:enforcer1}.

\begin{obs}\label{ob:enforcer2}
Consider the Boolean formula $x^{(2)}$ over a set $V$ of eight variables, where $x,y\in V$. A truth assignment $\beta$ for $\{x\}$ can be extended to a truth assignment $\beta'$ for $V$ that satisfies $x^{(2)}$ if and only if $\beta(x) = T$.
\end{obs}

\noindent We will use the $\mathcal{S}$-enforcer and $x^{(2)}$ as well as extensions thereof in the proofs of several results established in this paper. \\

\section{Hardness of balanced and unbalanced versions of \sc{$\forall \exists$ 3-SAT-$(s_1, s_2, t_1, t_2)$}} 

\subsection{New enforcers}\label{sec:enforcer}
We start by describing three new enforcers, with the first one building upon the enforcers introduced in the previous section. Consider the following gadget:
\begin{align*}
E(x) = \mathcal{S}(x, y, y) \wedge \mathcal{S}(x, \bar{y}, \bar{y}) \wedge \mathcal{S}(\bar{x}, z, \bar{z}) \wedge \mathcal{S}(z, \bar{z}, u) \wedge \mathcal{S}(u, \bar{u}, \bar{u})
\end{align*}
which is an extended variant of the enforcer $x^{(2)}$. We call $x$ the {\it enforcer variable} of $E(x)$. Note that every  variable in $\{u,y,z\}$ appears exactly twice unnegated and exactly twice negated in $E(x)$, and that $x$ appears exactly twice unnegated and exactly once negated in $E(x)$. Moreover, by construction and Observation~\ref{ob:enforcer2}, it follows that $E(x)$ is satisfiable by setting $x$ to be $T$, and by setting all remaining 18 variables appropriately. 
\begin{obs}\label{ob:E(x)}
Consider the gadget $E(x)$, and let $V$ be its associated set of variables. A truth assignment $\beta$ for $\{x\}$ can be extended to a truth assignment $\beta'$ for $V$ that satisfies $E(x)$ if and only if $\beta(x) = T$.
\end{obs}

We now turn to two  quantified enforcers whose purpose is to increase the number of universal variables by one and three, respectively, relative to the number of existential variables. First, let 
\begin{align*}
Q^{1} = &(u \vee v \vee a) \wedge (u \vee v \vee b) \wedge (\overline{u} \vee \overline{v} \vee \overline{a}) \wedge (\overline{u} \vee \overline{v} \vee \overline{b}) \wedge {} \\
&(a \vee \overline{b} \vee r) \wedge (\overline{a} \vee b \vee r) \wedge (c \vee \overline{d} \vee \overline{r}) \wedge (\overline{c} \vee d \vee \overline{r}) \wedge {} \\
&(w \vee q \vee c) \wedge (w \vee q \vee d) \wedge (\overline{w} \vee \overline{q} \vee \overline{c}) \wedge (\overline{w} \vee \overline{q} \vee \overline{d}),
\end{align*}
 where $u, v, w, q, r$ are universal variables, and $a, b, c, d$ are existential variables. Observe that each variable of $Q^1$ appears exactly twice unnegated and exactly twice negated. Second, let 
\begin{align*}
Q^3 = &(u \vee r \vee a) \wedge (\overline{u} \vee \overline{b} \vee \overline{a}) \wedge (v \vee q \vee b) \wedge {} \\
&(\overline{v} \vee \overline{r} \vee \overline{a})  \wedge (w \vee a \vee b) \wedge (\overline{w} \vee \overline{q} \vee \overline{b}), 
\end{align*}
where $u, v, w, q, r$ are  universal variables and $a, b$ are  existential variables. Observe that each universal variable of $Q^3$ appears exactly once unnegated and exactly once negated, and that each existential variable of $Q^3$ appears exactly twice unnegated and exactly twice negated.

\begin{lem}\label{lem:q_forall}
The quantified Boolean formula $$\forall u \,\forall v \, \forall w \, \forall q \, \forall r \, \exists a \, \exists b \, \exists c \, \exists d \mbox{ }Q^1$$
is a yes-instance of \sc{$\forall \exists$ 3-SAT}.
\end{lem}

\begin{proof}
Let $U = \{u, v, w, q, r\}$, and let $E = \{a, b, c, d\}$. Furthermore, let $\beta \colon U \rightarrow \{T, F\}$ be a truth assignment for $U$. We extend $\beta$ to $\beta' \colon U \cup E \rightarrow \{T, F\}$ as follows: 
\[
\beta'(a) = \beta'(b) = \overline{\beta(u)}, \quad \beta'(c) = \beta'(d) = \overline{\beta(w)}. 
\]
It is now easy to verify that $\beta'$ satisfies all clauses. Thus, $Q^1$ is a yes-instance of  \sc{$\forall \exists$ 3-SAT}.
\end{proof}

\begin{lem}\label{lem:q_forall-3}
The quantified Boolean formula $$\forall u \,\forall v \, \forall w \, \forall q \, \forall r \, \exists a \, \exists b \mbox{ }Q^3$$
is a yes-instance of \sc{$\forall \exists$ 3-SAT}.
\end{lem}

\begin{proof}
Let $U = \{u, v, w, q,r\}$, and let $E = \{a, b\}$. Futhermore, let $\beta \colon U \rightarrow \{T, F\}$ be a truth assignment for $U$. We extend $\beta$ to a truth assignment $\beta' \colon U \cup E \rightarrow \{T, F\}$ for $Q^3$ as follows: 
\[
\beta'(a) = \overline{\beta(u)} \wedge \overline{\beta(r)}, \quad \beta'(b) = \overline{\beta(w)} \vee \overline{\beta(q)}. 
\]
It is now straightforward to check that $Q^3$ is a yes-instance of  \sc{$\forall \exists$ 3-SAT}.
\end{proof}

\subsection{Hardness of {\sc Balanced $\forall \exists$ 3-SAT-$(s_1, s_2, t_1, t_2)$}}

In this section, we show that {\sc Balanced $\forall \exists$ 3-SAT-$(s_1, s_2, t_1, t_2)$} is $\Pi_2^P$-complete when
$$(s_1, s_2, t_1, t_2)\in \{(2, 2, 2, 2), (1, 1, 2, 2)\}.$$  
To this end, for a clause $C$, we use $\overline{C}$ to denote the clause obtained from $C$ by replacing each literal with its negation and call $\overline{C}$ the {\it complement} of $C$. For example, if $C=(x_1\vee\overline{x}_2\vee\overline{x}_3)$, then $\overline{C}=(\overline{x}_1\vee x_2\vee x_3)$.

\begin{thm}\label{thm:balanced}
{\sc Balanced $\forall \exists$ 3-SAT-$(2, 2, 2, 2)$} is $\Pi_2^P$-complete.  
\end{thm}

\begin{proof}
Noting that {\sc Balanced $\forall \exists$ 3-SAT-$(2, 2, 2, 2)$} is a special case of $\forall \exists$ 3-SAT, we deduce that the problem is in $\Pi_2^P$. We show that the problem is $\Pi_2^P$-hard by a reduction from {\sc $\forall \exists$ NAE-3-SAT}. Let 
\[
\Phi_1 = \forall X_1^p \exists X_{p+1}^n \varphi
\]
be an instance of $\forall\exists$ {\sc NAE-3-SAT}, where $$\varphi=\bigwedge_{j = 1}^m C_j$$ is a Boolean formula over a set $V_1=\{x_1,x_2,\ldots,x_n\}$ of variables such that $C_j$ is a disjunction of at most three literals for all $j\in \{1, 2, \ldots, m\}$ and no~$C_j$ contains a single literal since, otherwise, $\Phi_1$ is a no-instance.  Following Schaefer~\cite[p.\,298]{schaefer01} and noting that his reduction translates without changes to {\sc $\forall \exists$ NAE-3-SAT}, we first 
modify $\Phi_1$ using the following transformation that turns every universal variable $x_i$ of $\Phi_1$ into an existential variable $y_i$ and introduces the set of new universal variables $\{z_1,z_2,\ldots,z_p\}$: 
\begin{eqnarray*}
\Phi_2 &=& \forall Z_1^p \exists X_{p+1}^n \exists Y_{1}^p  \varphi[x_1 \mapsto y_1, \ldots, x_p \mapsto y_p] \wedge \bigwedge_{i=1}^p \left ((\bar{z}_i \vee y_i) \wedge (z_i \vee \bar{y}_i) \right ).  \\
&=&\forall Z_1^p \exists X_{p+1}^n \exists Y_{1}^p  \varphi',
\end{eqnarray*}
where $\varphi' = \varphi[x_1 \mapsto y_1, \ldots, x_p \mapsto y_p] \wedge \bigwedge_{i=1}^p \left ((\bar{z}_i \vee y_i) \wedge (z_i \vee \bar{y}_i) \right )$. Let $V_2$ be the set of variables of $\Phi_2$. 

\begin{sublemma}
$\Phi_1$ is a yes-instance of  {\sc $\forall \exists$ NAE-3-SAT} if and only if $\Phi_2$ is a yes-instance of  {\sc $\forall \exists$ NAE-3-SAT}. 
\end{sublemma}

\begin{proof}
First, suppose that $\Phi_1$ is a yes-instance of  {\sc $\forall \exists$ NAE-3-SAT}. Let $\beta_1$ be a truth assignment for $V_1$ that nae-satisfies $\Phi_1$, and let $\beta_2$ be the following truth assignment for $V_2$:
\begin{enumerate}[label=(\roman*), noitemsep]
\item set $\beta_2(x_i)=\beta_1(x_i)$ for each $i\in\{p+1,p+2,\ldots,n\}$;
\item set $\beta_2(y_i)=\beta_1(x_i)$ for each $i\in\{1,2,\ldots,p\}$;
\item set $\beta_2(z_i)=\beta_2(y_i)$ for each $i\in\{1,2,\ldots,p\}$.
\end{enumerate}
By (iii), it is straightforward to check that $\beta_2$ nae-satisfies $\Phi_2$ and that, for every truth assignment for $Z_1^p$,  there exists a truth assignment for $X_{p+1}^n\cup Y_1^p$ that nae-satisfies $\Phi_2$. Hence, $\Phi_2$ is a yes-instance of {\sc $\forall \exists$ NAE-3-SAT}.

Second, suppose that $\Phi_2$ is a yes-instance of  {\sc $\forall \exists$ NAE-3-SAT}. Let $\beta_2$ be a truth assignment for $V_2$ that nae-satisfies $\Phi_2$. By construction of $\Phi_2$, it follows that $\beta_2(z_i)=\beta_2(y_i)$ for each $i\in\{1,2,\ldots,p\}$. Hence $\beta_1$ with $\beta_1(x_i)=\beta_2(y_i)$ for each $i\in\{1,2,\ldots,p\}$, and $\beta_1(x_i)=\beta_2(x_i)$ for each $i\in\{p+1,p+2,\ldots,n\}$ is a truth assignment for $V_1$ that nae-satisfies $\Phi_2$. Thus $\Phi_1$ is a yes-instance of  {\sc $\forall \exists$ NAE-3-SAT}
\end{proof}

For each $w_i\in X_{p+1}^n\cup Y_1^p$, we use $a(w_i)$ to denote the number of appearances of $w_i$ in $\varphi'$ throughout the remainder of this proof. Next, we apply, in turn, the following transformation adapted from Berman et al.~\cite[p.\,4]{berman03} yielding an instance of \textsc{$\forall \exists$ 3-SAT}. 

\begin{enumerate}
\item Replace $\exists X_{p+1}^n$ in  $\Phi_2$ with the following list of existential variables: $$\exists x_{p+1,1}\exists x_{p+1,2}\cdots \exists x_{p+1,a(x_{p+1})}\cdots\exists x_{n,1}\exists x_{n,2}\cdots \exists x_{n,a(x_n)}$$ 
Similarly, replace $ \exists Y_{1}^p$ in $\Phi_2$ with the following list of existential variables: $$\exists y_{1,1}\exists y_{1,2}\cdots \exists y_{1,a(y_1)}\cdots \exists y_{p,1}\exists y_{p,2}\cdots \exists y_{p,a(y_p)}.$$
Lastly,  for each existential variable $w_i\in X_{p+1}^n\cup Y_1^p$ and all $k\in \{1, \dots, a(w_i)\}$, replace the $k$-th appearance of $w_i$ in $\varphi'$ by $w_{i,k}$.
\item Replace each clause $C_j$ with $C_j\wedge \overline{C}_j$. 
\item For each $w_i\in X_{p+1}^n\cup Y_1^p$, introduce the clauses 
\[
(\overline{w_{i,1}} \vee w_{i,2}) \wedge (\overline{w_{i,2}} \vee w_{i,3}) \wedge \cdots \wedge (\overline{w_{i,a(w_i)-1}} \vee w_{i,a(w_i)}) \wedge (\overline{w _{i,a(w_i)}}\vee w_{i,1}).
\]   
\item Replace each 2-clause $(\ell_1\vee \ell_2)$  by $(\ell_1\vee \ell_2\vee\overline{u})\wedge E(u)$, where $u$ and all 18 variables introduced by $E(u)$ are new existential variables. Append all 19 new variables to the list of existential variables. 
\end{enumerate}   
Let $\Phi_3$ denote the formula constructed by the preceding four-step procedure, and let $V_3$ be the set of variables of $\Phi_3$. 

\begin{sublemma}
$\Phi_2$ is a yes-instance of  {\sc $\forall \exists$ NAE-3-SAT} if and only if $\Phi_3$ is a yes-instance of  {\sc $\forall \exists$ 3-SAT}. 
\end{sublemma}

\begin{proof}
First, suppose that $\Phi_2$ is a yes-instance of  {\sc $\forall \exists$ NAE-3-SAT}. Let $\beta_2$ be a truth assignment for $V_2$ that nae-satisfies $\Phi_2$. Obtain a truth assignment~$\beta_3$ for $V_3$ as follows:
\begin{enumerate}[label=(\roman*), noitemsep]
\item set $\beta_3(z_i)=\beta_2(z_i)$ for each $i\in\{1,2,\ldots,p\}$;
\item set $\beta_3(x_{i,k})=\beta_2(x_i)$ for each $i\in\{p+1,p+2,\ldots,n\}$ and $k\in\{1,2,\ldots,a(x_i)\}$;
\item set $\beta_3(y_{i,k})=\beta_2(y_i)$ for each $i\in\{1,2,\ldots,p\}$ and $k\in\{1,2,\ldots,a(y_i)\}$.
\end{enumerate}
Additionally, for each 2-clause $C=(\ell_1\vee \ell_2)$  that is replaced with  the $\mathcal{S}$-enforcer $(\ell_1\vee \ell_2\vee\overline{u})\wedge E(u)$ in Step 4, set $\beta_3(u)=T$, and set all 18 existential variables introduced by $E(u)$ such that the 25 clauses of $E(u)$ are satisfied. By construction of $E(u)$ and Observation~\ref{ob:E(x)}, this is always possible. If $C$ is a 2-clause of $\Phi_2$, then, as $C$ is nae-satisfied by $\beta_2$, it follows that $\beta_3$ satisfies $(\ell_1\vee \ell_2\vee\overline{u})\wedge E(u)$. If $C$ is initially a 2-clause introduced in Step 3 and then replaced in Step 4, it follows by (ii) and (iii) that $\beta_3$ satisfies $(\ell_1\vee \ell_2\vee\overline{u})\wedge E(u)$.  Noting that if a truth assignment nae-satisfies a clause, then it also nae-satisfies its complement, it is now straightforward to check that $\beta_3$ satisfies $\Phi_3$ and, hence, $\Phi_3$ is a yes-instance of  {\sc $\forall \exists$ 3-SAT}.

Second, suppose that $\Phi_3$ is a yes-instance of  {\sc $\forall \exists$ 3-SAT}. Let $\beta_3$ be a truth assignment that satisfies $\Phi_3$. Let $u$ be an enforcer variable such that the 25 clauses associated with $E(u)$ are clauses of $\Phi_3$ but not of $\Phi_2$. By construction of $E(u)$ and Observation~\ref{ob:E(x)}, we have $\beta_3(u)=T$.  Now let $w_i\in X_{p+1}^n\cup Y_1^p$. As $\beta_3$ satisfies $\Phi_3$ and each enforcer variable that is contained in~$V_3$ is assigned to $T$ under $\beta_3$, it follows from the clauses introduced in Step~3 that $$\beta_3(w_{i,1})=\beta_3(w_{i,2})=\cdots=\beta_3(w_{i,{a(w_i)}}).$$ Let $\beta_2$ be the truth assignment for $\Phi_2$ that is obtained from $\beta_3$ as follows:
\begin{enumerate}[label=(\roman*), noitemsep]
\item $\beta_2(z_i)=\beta_3(z_i)$ for each $i\in\{1,2,\ldots,p\}$,
\item $\beta_2(x_i)=\beta_3(x_{i,1})$ for each $i\in\{p+1,p+2,\ldots,n\}$, and
\item $\beta_2(y_i)=\beta_3(y_{i,1})$ for each $i\in\{1,2,\ldots,p\}$.
\end{enumerate}
As $\beta_3$ satisfies $\Phi_3$, it immediately follows that $\beta_2$ satisfies $\Phi_2$. We complete the proof by showing that $\beta_2$ nae-satisfies $\Phi_2$. Assume that there exists a clause $C$ in $\Phi_2$ whose literals all evaluate to $T$ under $\beta_2$. Let $D$ be the clause obtained from $C$ by applying Step 1. If~$C$ contains exactly three literals, then all three literals of~$\overline{D}$ evaluate to~$F$; thereby contradicting that~$\beta_3$ satisfies $\Phi_3$. On the other hand, if~$C$ contains exactly two literals, then~$D$ is replaced with a 3-clause, say $D'$, and an enforcer, say $E(u')$, in Step 4 and, similarly, $\overline{D}$ is replaced with a 3-clause, say $D''$, and an enforcer, say~$E(u'')$, in Step~4. Note that $D''$ is not the complement of $D'$. Furthermore, again by Observation~\ref{ob:E(x)}, we have $\beta_3(u')=\beta_3(u'')=T$. Now, as each literal of $C$ evaluates to $T$,  each literal of $D''$ evaluates to $F$ under $\beta_3$; a contradiction. Hence $\beta_2$ nae-satisfies $\Phi_2$, and so $\Phi_2$ is a yes-instance of {\sc $\forall \exists$ NAE-3-SAT}.
\end{proof}

We next obtain a quantified Boolean formula $\Phi_4$ from $\Phi_3$ such that the number of universal variables in $\Phi_4$ is equal to the number of existential variables in $\Phi_4$. Let $p_e$  be the number of existential variables in $V_3$, and let~$p_u$ be the number of universal variables in $V_3$. By construction, observe that $p_u=p\geq 0$. Since a new existential variable $y_i$ has been introduced for each universal variable $x_i$ in $V_1$ with $i\in\{1,2,\ldots,p\}$, we  have $p_e\geq p_u$. Let $Q_k^1$ be the enforcer with variables $\{a_k,b_k,c_k,d_k,q_k,r_k,u_k,v_k,w_k\}$ as introduced in Section~\ref{sec:enforcer}. Obtain $\Phi_4$ from $\Phi_3$ by adding $Q_k^1$ to the boolean formula, appending $\exists a_k\exists b_k\exists c_k\exists d_k$ to the list of existential variables, and appending $\forall q_k\forall r_k\forall u_k\forall v_k\forall w_k$ to the list of universal variables for each $k\in\{1,2,\ldots, p_e-p_u\}$. It now follows that $\Phi_4$ contains $p_e+4(p_e-p_u)=5p_e-4p_u$ existential variables and $p_u+5(p_e-p_u)=5p_e-4p_u$ universal variables. Moreover, by Lemma~\ref{lem:q_forall}, we have that $\Phi_3$ is a yes-instance of {\sc $\forall \exists$ 3-SAT} if and only if $\Phi_4$ is a yes-instance of {\sc $\forall \exists$ 3-SAT}. 

We complete the proof by showing that $\Phi_4$ is an instance of {\sc Balanced $\forall \exists$ 3-SAT-$(2, 2, 2, 2)$}. Let $V_4$ be the set of variables of $\Phi_4$. By the transformation of $\Phi_1$ into $\Phi_3$ and the construction of $Q_k^1$, it is easily checked that each universal variable in $V_4$ appears exactly twice unnegated and exactly twice negated in $\Phi_4$. Now, consider the following three sets of existential variables: 
\begin{enumerate}[label=(\Roman*), noitemsep]
\item $S_1= \bigcup_{k=1}^{p_e-p_u}\{a_k,b_k,c_k,d_k\},$
\item $S_2=\bigcup_{i=p+1}^n\{x_{i,1},x_{i,2},\ldots,x_{i,a(x_i)}\}$ and
\item $S_3=\bigcup_{i=1}^p\{y_{i,1},y_{i,2},\ldots,y_{i,a(y_i)}\}$.
\end{enumerate}
It follows again from the construction of $Q_k^1$ that each variable in $S_1$ appears exactly twice unnegated and exactly twice negated in $\Phi_4$. Furthermore, by Steps 1--3 in the construction of $\Phi_3$, it follows that each variable in $S_2\cup S_3$ appears exactly twice unnegated and exactly twice negated in $\Phi_4$. Lastly, each existential variable in $V_4-(S_1 \cup S_2\cup S_3)$ has been introduced by replacing a 2-clause $(\ell_1\vee\ell_2)$ with $(\ell_1\vee\ell_2\vee\overline{u})\wedge E(u)$ in Step~4 of the construction of $\Phi_3$. Recall that  $u$ appears unnegated exactly twice and negated exactly once in $E(u)$, and that each of the 18 remaining variables introduced by $E(u)$ appears exactly twice unnegated and exactly twice negated in $E(u)$. It now follows that $\Phi_4$ is an instance of {\sc Balanced $\forall \exists$ 3-SAT-$(2, 2, 2, 2)$}. We complete the proof of this theorem by noting that each clause of $\Phi_4$ is a 3-clause that contains three distinct variables and that the size of $\Phi_4$ is polynomial in the size of $\Phi_1$.
\end{proof}

In Theorem~\ref{thm:balanced}, we have imposed the same bound on existential and universal variables, i.e. $s_1=s_2=t_1=t_2$. By allowing separate bounds, i.e. $s_1=s_2$ and $t_1=t_2$, we obtain the following stronger result. 

\begin{thm}\label{thm:balanced_sep}
{\sc Balanced $\forall \exists$ 3-SAT-$(1, 1, 2, 2)$} is $\Pi_2^P$-complete.
\end{thm}

\begin{proof}
Clearly, {\sc Balanced $\forall \exists$ 3-SAT-$(1, 1, 2, 2)$} is in $\Pi_2^P$. We establish $\Pi_2^P$-hardness via a reduction from {\sc Balanced $\forall \exists$ 3-SAT-$(2, 2, 2, 2)$}. Let
\[
\Phi_1 = \forall X_1^p \exists X_{p+1}^{2p} \varphi. 
\]
be an instance of {\sc Balanced $\forall \exists$ 3-SAT-$(2, 2, 2, 2)$}. Let $m$ be the number of 3-clauses of $\varphi$. As $3m=2p\cdot 4$, observe that $p$ is divisible by 3. Following a similar strategy as in the proof of Theorem~\ref{thm:balanced}, we apply the following 4-step process to transform $\Phi_1$ into an instance~$\Phi_4$ of {\sc Balanced $\forall \exists$ 3-SAT-$(1, 1, 2, 2)$}.
\begin{enumerate}
\item Obtain
\begin{align*}
\Phi_2 = \forall C_1^p \exists X_{p+1}^{2p} \exists Y_{1}^p \exists Z_{1}^{6p}  &\varphi[x_1 \mapsto y_1, \ldots, x_p \mapsto y_p] \wedge {} \\
&\bigwedge_{i=1}^p \left (\mathcal{S}_u(\overline{c_i}, y_i, y_i) \wedge \mathcal{S}_u(c_i, \overline{y_i}, \overline{y_i}) \right ), 
\end{align*}
by turning each universal variable in $x_i\in X_1^p$ into an existential variable $y_i$, adding new universal variables $c_1,c_2,\ldots,c_p$, and adding new existential variables $z_1,z_2,\ldots,z_{6p}$ that are introduced as new variables by copies of the $\mathcal{S}$-enforcer. By construction, each $y_i\in Y_1^p$ appears exactly four times unnegated and exactly four times negated in $\Phi_2$.
\item For each $y_i\in Y_1^p$ and $k\in\{1,2,3,4\}$, replace the $k$-th negated appearance of $y_i$ with $\overline{y_{i,k}}$ and replace the $k$-th unnegated appearance of $y_i$ with $y_{i,k}$. Then replace $\exists Y_1^p$ in $\Phi_2$ with the following list of existential variables
$$\exists y_{1,1}\exists y_{1,2}\exists y_{1,3}\exists y_{1,4}\cdots \exists y_{p,1}\exists y_{p,2}\exists y_{p,3}\exists y_{p,4}.$$
\item Add the following clauses to the Boolean formula resulting from Step~2:
\begin{align*}
\bigwedge_{i = 1}^p \biggl [  &(\overline{y_{i,1}} \vee y_{i, 2} \vee \overline{d_{i, 1}}) \wedge (\overline{y_{i, 2}} \vee y_{i, 3} \vee \overline{d_{i, 1}}) \wedge d_{i, 1}^{(2)}  \wedge {} \\ 
&(\overline{y_{i,3}} \vee y_{i, 4} \vee \overline{d_{i, 2}}) \wedge (\overline{y_{i,4}} \vee y_{i, 1} \vee \overline{d_{i, 2}}) \wedge d_{i, 2}^{(2)} \biggr],
\end{align*}
where $d_{i, 1}$ and $d_{i, 2}$  are new existential variables with $i\in\{1,2\ldots,p\}$, and $d_{i, 1}^{(2)}$ and $d_{i, 2}^{(2)}$ are the corresponding enforcers as introduced in Section~\ref{sec:prelim}. Then append $$\exists d_{1,1}\exists d_{1,2}\exists d_{2,1}\exists d_{2,2}\cdots\exists d_{p,1}\exists d_{p,2}\exists E_1^{14p}$$ to the list of existential variables, where $E_1^{14p}$ is the set of new variables introduced by these enforcers (each of $d_{i, 1}^{(2)}$ and $d_{i, 2}^{(2)}$ introduces seven such variables). Let $\Phi_3$ denote the resulting quantified Boolean formula.
\item Note that each universal variable of $\Phi_3$ appears exactly once unnegated and exactly once negated, and that each existential variable of $\Phi_3$ appears exactly twice unnegated and exactly twice negated. Let $p_e$ (resp.\ $p_u$) be the number of existential (resp.\ universal) variables in~$\Phi_3$. Then $$p_e = p + 4p+ 6p +  2p + 14p = 27p \quad \textnormal{ and }\quad p_u = p.$$ Evidently, $p_e\geq p_u$. Furthermore, as $p$ is divisible by $3$, it follows that $p_e$ and $p_u$ are both divisible by 3. Let $\Delta=(p_e-p_u)/3$. Now, for each $k\in\{1,2,\ldots, \Delta\}$, add the enforcer $Q_k^3$ as introduced in Section~\ref{sec:enforcer} to $\Phi_3$, append $\exists a_k\exists b_k$ to the list of existential variables, and append $\forall q_k\forall r_k\forall u_k\forall v_k\forall w_k$ to the list of universal variables.
\end{enumerate}
Let $\Phi_4$ denote the formula resulting from the preceding 4-step process. By construction, each clause in $\Phi_4$ is a 3-clause that contains three distinct variables. Moreover, since, for each $k$, the enforcer $Q_k^3$ increases the number of universal variables by five and the number of existential variables by two, it follows that $\Phi_4$ is an instance of {\sc Balanced $\forall \exists$ 3-SAT-$(1, 1, 2, 2)$}. 

Noting that the size of $\Phi_4$ is polynomial in the size of $\Phi_1$, we complete the proof by establishing the following statement.

\begin{sublemma}\label{sublem:iff}
$\Phi_1$ is a yes-instance of {\sc Balanced $\forall \exists$ 3-SAT-$(2,2, 2, 2)$}  if and only if $\Phi_4$ is a yes-instance of {\sc Balanced $\forall \exists$ 3-SAT-$(1, 1, 2, 2)$}.
\end{sublemma}

\begin{proof}
Let  $V_1$ be the set of variables of $\Phi_1$, and let $V_4$ be the set of variables of $\Phi_4$. First, suppose that $\Phi_1$ is a yes-instance of {\sc Balanced $\forall \exists$ 3-SAT-$(2,2, 2, 2)$}. Let $\beta_1$ be a truth assignment that satisfies $\Phi_1$. We obtain a truth assignment $\beta_4$ for a subset of $V_4$, say $V_4'$, from $\beta_1$ as follows:
\begin{enumerate}[label=(\roman*), noitemsep]
\item for each $i\in\{1,2,\ldots,p\}$, set $\beta_4(c_i)=\beta_1(x_i)$;
\item for each $i\in\{p+1,p+2,\ldots,2p\}$, set $\beta_4(x_i)=\beta_1(x_i)$;
\item for each $i\in\{1,2,\ldots,p\}$ and $k\in\{1,2,3,4\}$, set $\beta_4(y_{i,k})=\beta_4(c_i)$;
\item for each $i\in\{1,2,\ldots,p\}$ and $k\in\{1,2\}$, set $\beta_4(d_{i,k})=T$.
\end{enumerate}
It is straightforward to check that each clause in $\Phi_4$ that does not contain a variable in  $$(A_1^\Delta\cup B_1^\Delta\cup E_1^{14p}\cup Q_1^\Delta\cup R_1^\Delta\cup U_1^\Delta\cup V_1^\Delta\cup W_1^\Delta \cup Z_1^{6p})$$ is satisfied by $\beta_4$. We next extend $\beta_4$ in three steps. First, by (iv) and Observation~\ref{ob:enforcer2}, it follows that $\beta_4$  extends to $V_4'\cup E_1^{14p}$ such that, for each $i\in\{1,2,\ldots,p\}$, the clauses of $d_{i,1}^{(2)}$ and $d_{i,2}^{(2)}$ are satisfied. Second, by Lemma~\ref{lem:q_forall-3}, $\beta_4$ also extends to $$V_4'\cup A_1^\Delta\cup B_1^\Delta \cup Q_1^\Delta\cup R_1^\Delta\cup U_1^\Delta\cup V_1^\Delta\cup W_1^\Delta$$ such that each clause in $Q_1^3\wedge Q_2^3\wedge\cdots\wedge Q_\Delta^3$ is satisfied. Third, by (i), (iii),  and Observation~\ref{ob:enforcer1} together with its subsequent remark, it follows that $\beta_4$ extends to $V_4'\cup Z_1^{6p}$ such that the clauses in $$\bigwedge_{i=1}^p \left (\mathcal{S}_u(\overline{c_i}, y_i, y_i) \wedge \mathcal{S}_u(c_i, \overline{y_i}, \overline{y_i}) \right )$$ are satisfied. We deduce that $\Phi_4$ is satisfiable.

Second, suppose that $\Phi_4$ is a yes-instance of {\sc Balanced $\forall \exists$ 3-SAT-$(1,1, 2, 2)$}. Let $\beta_4$ be a truth assignment that satisfies $\Phi_4$. It follows from Observation~\ref{ob:enforcer2}, that $\beta_4(d_{i,1})=\beta_4(d_{i,2})=T$ for each $i\in\{1,2,\ldots,p\}$. Hence, the clauses introduced in Step 3 imply that $$\beta_4(y_{i,1})=\beta_4(y_{i,2})=\beta_4(y_{i,3})=\beta_4(y_{i,4}).$$ It is now easy to check that the truth assignment $\beta_1$ for $V_1$ obtained from $\beta_4$ by setting
\begin{enumerate}[label=(\roman*), noitemsep]
\item $\beta_1(x_i)=\beta_4(y_{i,1})$ for each $i\in\{1,2,\ldots,p\}$ and
\item $\beta_1(x_i)=\beta_4(x_i)$ for each $i\in\{p+1,p+2,\ldots,2p\}$
\end{enumerate}
satisfies $\Phi_1$. Thus, statement~\ref{sublem:iff} holds. 
\end{proof}
This completes the proof of Theorem~\ref{thm:balanced_sep}.
\end{proof}

We end this section by remarking that Haviv et al.~\cite[p.\,55]{haviv07} showed that {\sc $\forall \exists$ 3-SAT-$(s_1, s_2, t_1, t_2)$} is in NP if $s_1+s_2\leq 1$ and in co-NP if $t_1+t_2\leq 2$. The latter result implies that  {\sc Balanced $\forall \exists$ 3-SAT-$(1,1,1,1)$} is in co-NP. Hence, unless the polynomial hierarchy collapses, the  balanced bounds on the number of appearances of universal and existential variables established in Theorems~\ref{thm:balanced} and~\ref{thm:balanced_sep} are the best possible ones (i.e., for smaller values, the problems can be placed on a lower level of the polynomial hierarchy). 

\subsection{Hardness of {\sc $\forall \exists$ 3-SAT-$(s_1, s_2, t_1, t_2)$}}

Following on from the results by Haviv et al.~\cite[p.\,55]{haviv07} mentioned in the last paragraph, {\sc $\forall \exists$ 3-SAT-$(s_1, s_2, t_1, t_2)$} with $s_1+s_2\leq 1$ or $t_1+t_2\leq 2$ is not $\Pi_2^P$-hard unless the polynomial hierarchy collapses. In this section, we show which instances of {\sc $\forall \exists$ 3-SAT-$(s_1, s_2, t_1, t_2)$} are NP-complete and which are $\Pi_2^P$-complete. Specifically, we show that {\sc $\forall \exists$ 3-SAT-$(s_1,s_2,t_1,t_2)$} is NP-complete for when $s_1+s_2=1$ and $(t_1,t_2)\in \{(1,2),(2,1)\}$, and $\Pi_2^P$-complete for when $s_1=s_2=1$ and $(t_1,t_2)\in \{(1,2),(2,1)\}$.

Let $\Phi$ be an instance of {\sc $\forall \exists$ 3-SAT-$(s_1,s_2,t_1,t_2)$} with $s_1+s_2=1$, and let $Y_1^p$ be the set of universal variables of $\Phi$. As noted by Haviv et al.~\cite[p.\,55]{haviv07}, we can obtain an equivalent unquantified Boolean formula from $\Phi$ by deleting all literals in $\{y_i,\overline{y}_i:i\in\{1,2,\ldots,p\}\}$ in the clauses of $\Phi$. Hence,  if $\Phi$ has the additional property that $t_1+t_2\leq 2$, it follows from results by Tovey~\cite[Section 3]{tovey84} that it can be decided in polynomial time whether or not $\Phi$ is a yes-instance. Hence, {\sc $\forall \exists$ 3-SAT-$(s_1,s_2,t_1,t_2)$} with $s_1+s_2=1$ and $t_1+t_2\leq 2$ is  polynomial-time solvable. The next theorem shows that {\sc $\forall \exists$ 3-SAT-$(s_1,s_2,t_1,t_2)$} with $s_1+s_2=1$ becomes NP-complete if $(t_1,t_2)\in \{(1,2),(2,1)\}$. To establish this result, we use  a variant of 3-SAT in which each clause is either a 2-clause or a 3-clause, and each variable appears exactly twice unnegated and exactly once negated, or exactly once unnegated and exactly twice negated. We refer to this variant as {\sc 3-SAT-$(3)$}. It was shown by Dahlhaus et al.~\cite[p.\,877f]{dahlhaus94} that {\sc 3-SAT-$(3)$} is NP-complete. To establish the next theorem, we impose the following two restrictions on an instance $\varphi$ of {\sc 3-SAT-$(3)$}.
\begin{enumerate}[label=(R\arabic*), noitemsep]
\item Each 2-clause (resp. 3-clause) contains 2 (resp. 3) distinct variables.
\item Amongst the clauses, each variable appears exactly twice unnegated and exactly once negated.
\end{enumerate}
Indeed, it follows immediately from Dahlhaus et al's.~\cite[p.\,877f]{dahlhaus94} construction that $\varphi$ satisfies (R1). Moreover,  standard pre-processing that replaces each literal of a variable that appears exactly once unnegated and exactly twice negated with its negation can be used to obtain an instance $\varphi'$ from $\varphi$ that satisfies (R2) and that is equivalent to $\varphi$. We hence obtain the following theorem.

\begin{thm}\label{thm:AE3Sat_NPc}
{\sc $\forall \exists$ 3-SAT-$(s_1,s_2,t_1,t_2)$} is {\rm NP}-complete if $s_1+s_2=1$ and $(t_1,t_2)\in \{(1,2),(2,1)\}$.   
\end{thm} 
\begin{proof}
It was shown by Haviv et al.~\cite[p.\,55]{haviv07} that {\sc $\forall \exists$ 3-SAT-$(s_1,s_2,t_1,t_2)$} with $s_1+s_2=1$ is in NP. We first establish NP-completeness for  {\sc $\forall \exists$ 3-SAT-$(1,0,2,1)$} via a reduction from  {\sc 3-SAT-$(3)$}.

Let
\[
\varphi = \bigwedge_{j=1}^p C_j^2 \wedge \bigwedge_{j=p+1}^m C_j^3
\]
be an instance of  {\sc 3-SAT-$(3)$} over a set $X^n_1$ of variables and such that each clause $C_j^k$ is a $k$-clause with $k\in\{2,3\}$. As described prior to the statement of Theorem~\ref{thm:AE3Sat_NPc}, we may assume that $\varphi$ satisfies (R1) and (R2). Construct the following quantified Boolean formula~$\Phi$ from $\varphi$:
\[
\Phi = \forall Y_1^p \exists X_1^n \left (\bigwedge_{j = 1}^p (C_j^2 \vee y_i) \wedge \bigwedge_{j=p+1}^m C_j^3 \right ). 
\]
Since $\varphi$ satisfies (R2), $\Phi$ is an instance of {\sc $\forall \exists$ 3-SAT-$(1,0,2,1)$}.
First, suppose that $\varphi$ is satisfiable. Then there is a truth assignment $\beta$ that satisfies each clause in~$\varphi$. In particular, $\beta$  satisfies each clause $C_j^2$ and, hence, any extension of $\beta$ to $Y_1^p$with $i \in \{1, 2, \ldots, p\}$ is a truth assignment that satisfies~$\Phi$. Second, suppose that $\Phi$ is satisfiable. Let $\beta'$ be a truth assignment for~$\Phi$ such that $\beta'(y_i)=F$ for each $i\in\{1,2,\ldots,p\}$. By the existence of $\beta'$, it follows that $\beta(x_i)=\beta'(x_i)$ for each $i\in\{1,2,\ldots,n\}$ is a truth assignment that satisfies each clause in $\varphi$. 
As the size of $\Phi$ is polynomial in the size of $\varphi$, NP-completeness of {\sc $\forall \exists$ 3-SAT-$(1,0,2,1)$} now follows. To see that {\sc $\forall \exists$ 3-SAT-$(s_1,s_2,t_1,t_2)$} is also NP-complete for when  
\begin{enumerate}[label=(\Roman*), noitemsep]
\item $(s_1,s_2,t_1,t_2)=(0,1,2,1)$,
\item $(s_1,s_2,t_1,t_2)=(1,0,1,2)$, or
\item $(s_1,s_2,t_1,t_2)=(0,1,1,2)$,
\end{enumerate} 
observe that an argument analogous to the one above applies if each literal of a universal variable in $\Phi$ is replaced with its negation to establish (I), if each literal of an existential variable in $\Phi$ is replaced with its negation to establish (II), and if each literal in $\Phi$ is replaced with its negation to establish (III).
\end{proof}

\begin{thm}\label{thm:AE3Sat_P2P}
{\sc $\forall \exists$ 3-SAT-$(1,1,t_1,t_2)$} with $(t_1,t_2)\in\{(1,2),(2,1)\}$  is $\Pi_2^P$-complete.
\end{thm}

\begin{proof}
We first establish the theorem for $(t_1,t_2)=(2,1)$. Throughout the proof, we make use of the following quantified enforcer for an existential variable $d_{i,k}$:
\[
E_\forall(d_{i,k}) = (d_{i,k} \vee u_{i,k} \vee v_{i,k} )\wedge (d_{i,k} \vee \overline{u_{i,k}} \vee \overline{v_{i,k}}),
\]
where $u_{i,k}$ and $v_{i,k}$ are new  universal variables for some $i,k\in\mathbb{Z}^+$. The following property of $E_\forall(d_{i,k})$ is easy to verify.
\begin{description}
\item (P) The Boolean formula $\forall u_{i,k}\forall v_{i,k} \exists d_{i,k}E_\forall(d_{i,k})$ is a yes-instance of {\sc $\forall \exists$ 3-SAT}. In particular, if a truth assignment $\beta$ for $\{d_{i,k},u_{i,k},v_{i,k}\}$ has the property that $\beta(d_{i,k})=T$, then $\beta$ satisfies  $E_\forall(d_{i,k})$. Furthermore, if $\beta$ satisfies $E_\forall(d_{i,k})$ and $\beta(u_{i,k})=\beta(v_{i,k})$, then this implies that $\beta(d_{i,k})=T$.
\end{description}

As {\sc $\forall \exists$ 3-SAT-$(1,1,2,1)$} is a special case  of {\sc $\forall \exists$ 3-SAT}, it follows that the former problem is in $\Pi_2^P$. 
We show $\Pi_2^P$-hardness by a reduction from {\sc Balanced $\forall \exists$ 3-SAT-$(1,1,2,2)$}, for which $\Pi_2^P$-completeness was established in Theorem~\ref{thm:balanced_sep}. Let $$\Phi_1=\forall X_1^p\exists Y_{p+1}^{2p}\varphi$$ be an instance of {\sc Balanced $\forall \exists$ 3-SAT-$(1,1,2,2)$}. The reduction has two steps:
\begin{enumerate}
\item For each existential variable $y_i$ of $\Phi_1$ with $i\in\{p+1,p+2,\ldots,2p\}$, replace the first  (resp.\ second) unnegated appearance of $y_i$ with $y_{i,1}$ (resp. $y_{i,2}$), replace the first  (resp.\ second) negated appearance of $y_i$ with the negation of $y_{i,3}$ (resp. $y_{i,4}$), and  add the new clauses
\begin{align*}
&(\overline{y_{i,1}} \vee y_{i, 2} \vee \overline{d_{i, 1}}) \wedge E_\forall(d_{i,1}) \wedge (\overline{y_{i, 2}} \vee y_{i, 3} \vee \overline{d_{i, 2}}) \wedge E_\forall(d_{i,2})  \wedge {} \\ 
&(\overline{y_{i,3}} \vee y_{i, 4} \vee \overline{d_{i, 3}}) \wedge E_\forall(d_{i,3}) \wedge (\overline{y_{i,4}} \vee y_{i, 1} \vee \overline{d_{i, 4}}) \wedge E_\forall(d_{i,4}),
\end{align*}
to $\Phi_1$, where each $d_{i,k}$ with $k \in \{1, 2, 3, 4\}$ is a new existential variable. For all $i\in \{p+1, p+2, \ldots, 2p\}$ and $k\in \{1, 2, 3, 4\}$, append $y_{i, k}$ and $d_{i,k}$ to the list of existential variables and append $u_{i, k}$ and $v_{i, k}$ to the list of universal variables. 
\item For each existential variable $y_{i,k}$ with $k\in\{3,4\}$, replace each literal $y_{i,k}$ with $\overline{y_{i,k}}$ and each literal $\overline{y_{i,k}}$ with $y_{i,k}$. 
\end{enumerate}
Let $\Phi_2$ be the resulting quantified Boolean formula, and let $V_2$ be the set of variables of $\Phi_2$. Note that each existential variable $y_{i,k}$ with $k\in\{3,4\}$  appears exactly once unnegated and exactly twice negated in the Boolean formula resulting from Step 1. Hence, due to Step 2, it follows that $\Phi_2$ is an instance of {\sc $\forall \exists$ 3-SAT-$(1,1,2,1)$}. Furthermore, for each $i\in\{p+1,p+2,\ldots,2p\}$, the clauses introduced in Step 1 are replaced with the following clauses in the course of Step 2:
\begin{align*}
&(\overline{y_{i,1}} \vee y_{i, 2} \vee \overline{d_{i, 1}}) \wedge E_\forall(d_{i,1}) \wedge (\overline{y_{i, 2}} \vee \overline{y_{i, 3}} \vee \overline{d_{i, 2}}) \wedge E_\forall(d_{i,2})  \wedge {} \\ 
&(y_{i,3} \vee \overline{y_{i, 4}} \vee \overline{d_{i, 3}}) \wedge E_\forall(d_{i,3}) \wedge (y_{i,4} \vee y_{i, 1} \vee \overline{d_{i, 4}}) \wedge E_\forall(d_{i,4}).
\end{align*}

We complete the proof for $(t_1, t_2)=(2, 1)$ by establishing the following statement.

\begin{sublemma}\label{sublem:iff-2}
$\Phi_1$ is a yes-instance of {\sc Balanced $\forall \exists$ 3-SAT-$(1,1, 2, 2)$}  if and only if $\Phi_2$ is a yes-instance of {\sc $\forall \exists$ 3-SAT-$(1,1,2,1)$}.
\end{sublemma}

\begin{proof}
First, suppose that $\Phi_1$ is a yes-instance of {\sc Balanced $\forall \exists$ 3-SAT-$(1,1, 2, 2)$}. Let $\beta_1$ be a truth assignment that satisfies $\Phi_1$. For every truth assignment $\beta_2'$ for the universal variables in $$\{u_{i,k},v_{i,k}:i\in\{p+1,p+2,\ldots,2p\}\textnormal{ and } k\in\{1,2,3,4\}\},$$ we extend $\beta_2'$ to a truth assignment $\beta_2$ for $V_2$ as follows:
\begin{enumerate}[label=(\roman*), noitemsep]
\item set $\beta_2(x_i)=\beta_1(x_i)$ for each $i\in\{1,2,\ldots,p\}$;
\item set $\beta_2(y_{i,k})=\beta_1(y_i)$ for each $i\in\{p+1,p+2,\ldots,2p\}$ and $k\in\{1,2\}$;
\item set $\beta_2(y_{i,k})=\overline{\beta_1(y_i)}$ for each $i\in\{p+1,p+2,\ldots,2p\}$ and $k\in\{3,4\}$;
\item set $\beta_2(d_{i,k})=T$ for each $i\in\{p+1,p+2,\ldots,2p\}$ and $k\in\{1,2,3,4\}$.
\end{enumerate}
Due to (iv) and Property (P), it is now easily checked that $\Phi_2$ is a yes-instance of {\sc $\forall \exists$ 3-SAT-$(1,1,2,1)$}.

Second, suppose that $\Phi_2$ is a yes-instance of  {\sc $\forall \exists$ 3-SAT-$(1,1,2,1)$}. Let~$\beta_2$ be a truth assignment that satisfies $\Phi_2$ such that $\beta_2(u_{i, k})=\beta_2(v_{i, k})$ for each $i\in\{p+1,p+2,\ldots 2p\}$ and $k\in\{1,2,3,4\}$. Since $\Phi_2$ is a yes-instance, this implies that $\beta_2(d_{i,k})=T$ by Property (P). Moreover, by construction, we have $$\beta_2(y_{i,1})=\beta_2(y_{i,2})\quad\textnormal{ and }\quad \overline{\beta_2(y_{i,1})}=\beta_2(y_{i,3})=\beta_2(y_{i,4})$$ for each $i\in\{p+1,p+2,\ldots 2p\}$. It now follows that $\beta_1$ with
\begin{enumerate}[label=(\roman*), noitemsep]
\item $\beta_1(x_i)=\beta_2(x_i)$ for each $i\in\{1,2,\ldots,p\}$ and 
\item $\beta_1(y_{i})=\beta_2(y_{i,1})$ for each $i\in\{p+1,p+2,\ldots,2p\}$
\end{enumerate}
is a truth assignment for the set of variables of $\Phi_1$ that satisfies each clause in $\Phi_1$ and, thus, $\Phi_1$ is a yes-instance of {\sc Balanced $\forall \exists$ 3-SAT-$(1,1, 2, 2)$}.
\end{proof}
Noting that the size of $\Phi_2$ is polynomial in the size of $\Phi_1$, the theorem now follows for $(t_1,t_2)=(2,1)$. Moreover,  replacing $k\in\{3,4\}$ with $k\in\{1,2\}$ in Step 2 of the reduction and, subsequently, applying an argument that is analogous to~\ref{sublem:iff-2}, establishes the theorem for $(t_1,t_2)=(1,2)$.
\end{proof}

\section{Hardness of {\sc Monotone $\forall \exists$ NAE-3-SAT-$(s,t)$} with bounded variable appearances}

\subsection{Enforcers}\label{sec:enforcer-2}

In this section, we describe four monotone enforcers that have recently been introduced in an unquantified context by Darmann and D\"ocker~\cite{darmann19}. For the purposes of this section, we use their enforcers in a quantified setting. Specifically, for the first three enforcers, $x$ can be a universally or existentially quantified variable.

\noindent {\bf Auxiliary non-equality gadget.} First, consider the auxiliary non-equality gadget
\begin{align*}
\operatorname{NE_{aux}}(x, y) = &(x \vee y \vee a) \wedge (x \vee y \vee b) \wedge (a \vee b \vee u) \wedge {} \\ &(a \vee b \vee v) \wedge (a \vee b \vee w) \wedge (u \vee v \vee w), 
\end{align*}  
where $a, b, u, v, w$ are five new existential variables, $y$ is an existential variable, and $x$ is a universal or existential variable. To nae-satisfy the last clause, at least one variable in $\{u, v, w\}$ is set to be $T$ and at least one is set to be $F$. Then, by the three preceding clauses, we have that a truth assignment that nae-satisfies $\operatorname{NE_{aux}}(x, y)$ assigns different truth values to $a$ and $b$. The next observation  follows by construction of the first two clauses.

\begin{obs}\label{ob:NE_aux}
Consider the gadget $\operatorname{NE_{aux}}(x, y)$, and let $V$ be its associated set of variables. A truth assignment $\beta$ for $\{x,y\}$ can be extended to a truth assignment $\beta'$ for $V$ that nae-satisfies $\operatorname{NE_{aux}}(x, y)$ if and only if $\beta(x) \neq \beta(y)$.
\end{obs}

\noindent {\bf Equality gadget.} The second enforcer is the equality gadget
\[
\operatorname{EQ}(x, y) = \operatorname{NE_{aux}}(p, q) \wedge \operatorname{NE_{aux}}(p, r)  \wedge (x \vee q \vee r) \wedge (y \vee q \vee r),
\]  
where $p, q, r$ are three new existential variables, $y$ is an existential variable, and $x$ is a universal or existential variable. By construction and Observation~\ref{ob:NE_aux}, a truth assignment that nae-satisfies $\operatorname{EQ}(x, y)$ assigns the same truth value to $q$ and $r$. The next observation follows by construction of the last two clauses in the equality gadget.

\begin{obs}\label{ob:EQ}
Consider the gadget $\operatorname{EQ}(x, y)$, and let $V$ be its associated set of variables. A truth assignment $\beta$ for $\{x,y\}$ can be extended to a truth assignment $\beta'$ for $V$ that nae-satisfies $\operatorname{EQ}(x, y)$ if and only if $\beta(x) = \beta(y)$.
\end{obs}

\noindent {\bf Non-equality gadget.} Combining the first and second enforcer, we now obtain another non-equality gadget:
\[
\operatorname{NE}(x, y) = \operatorname{EQ}(x, p) \wedge \operatorname{EQ}(y, q) \wedge \operatorname{NE_{aux}}(p, q), 
\]  
where $p$ and $q$ are two new existential variables, $y$ is an existential variable, and $x$ is a universal or existential variable. The next observation follows immediately by construction, and Observations~\ref{ob:NE_aux} and \ref{ob:EQ}. 

\begin{obs}\label{ob:NE}
Consider the gadget $\operatorname{NE}(x, y)$, and let $V$ be its associated set of variables. A truth assignment $\beta$ for $\{x,y\}$ can be extended to a truth assignment $\beta'$ for $V$ that nae-satisfies $\operatorname{NE}(x, y)$ if and only if $\beta(x) \neq \beta(y)$.
\end{obs}

The next observation follows by construction of the last three enforcers. 
\begin{obs}\label{ob:varInc}
Let $\mathcal{E}$ be an enforcer in $\{\operatorname{NE_{aux}}(x, y), \operatorname{EQ}(x, y), \operatorname{NE}(x, y)\}$. Then each  variable introduced by $\mathcal{E}$ appears at most four times in $\mathcal{E}$.
\end{obs}

\noindent {\bf Padding gadget.} The fourth enforcer is the gadget
\begin{align*}
\operatorname{P1}(x) = &(x \vee a \vee b) \wedge (a \vee c \vee d) \wedge (a \vee b \vee e) \wedge {} \\ & (a \vee d \vee e) \wedge  (b \vee c \vee d) \wedge  (b \vee c \vee e) \wedge  (c \vee d \vee e),
\end{align*}
where $x$ is an existential variable, and $a, b, c, d, e$ are five new existential variables each of which appears exactly four times in the gadget. For the truth assignment $\beta$ for $\{a,b,c,d,e\}$ with $\beta(a)=\beta(c)=\beta(e)=T$ and $\beta(b)=\beta(d)=F$, the next observation follows immediately by construction.
\begin{obs}\label{ob:padding-gadget}
The gadget $\operatorname{P1}(x)$  is nae-satisfiable.~Moreover,~every~truth assignment for $\{x\}$ can be extended to a truth assignment for $\{a,b,c,d,e,x\}$ that nae-satisfies $\operatorname{P1}(x)$.
\end{obs}
\noindent Intuitively, ${P1}(x)$ is used to increase the number of appearances of existential variables in a Boolean formula until each variable appears exactly four times.

\subsection{Hardness of {\sc Monotone $\forall \exists$ NAE-3-SAT-$(s,t)$}} 
In this section, we establish that a monotone and linear Boolean formula $\varphi$ of $\forall \exists$ NAE-3-SAT is complete for the second level of the polynomial hierarchy even if each clause in $\varphi$ contains at most one universal variable and, amongst the clauses in $\varphi$, each universal universal variable appears exactly once and each existential variable appears exactly three times. We start by establishing a slightly weaker result without linearity.

\begin{proposition}\label{prop:nae}
{\sc Monotone $\forall \exists$ NAE-3-SAT-$(1, 4)$} is $\Pi_2^P$-complete if each clause contains at most one universal variable.
\end{proposition}

\begin{proof}
Clearly, the decision problem {\sc Monotone $\forall \exists$ NAE-3-SAT-$(1, 4)$} as described in the statement of the proposition is in $\Pi_2^P$. We show that it is $\Pi^P_2$-complete by a reduction from {\sc $\forall \exists$ NAE-3-SAT}. For the latter problem, $\Pi_2^P$-completeness was established by Eiter and Gottlob~\cite{eiter95}. Let
\[
\Phi_1 = \forall X_1^p \exists Y_{p+1}^n \varphi
\]
be an instance of {\sc $\forall \exists$ NAE-3-SAT} over a set $V_1 = X_1^p \cup Y_{p+1}^n$ of variables. We may assume that each clause contains exactly three literals and at most one duplicate literal.  

In what follows, we  construct two quantified Boolean formulas that include copies of the enforcers introduced in Section~\ref{sec:enforcer-2}. Each such enforcer adds several new existential variables. For ease of exposition throughout this proof, we use $A$ to denote the set of all new existential variables that are introduced by a copy of an enforcer in  $$\{\operatorname{NE_{aux}}(x, y), \operatorname{EQ}(x, y), \operatorname{NE}(x, y),\operatorname{P1}(x)\}.$$ In particular,  $A$ is initially empty and, each time we use a new enforcer copy, we  add the newly introduced variables to $A$ and append them to the list of existential variables without mentioning it explicitly. We remark that it will always be clear from the context that the number of elements in $A$ is polynomial in the size of $\Phi_1$.

Now, let 
\[
\Phi_2 = \forall Z_1^p \exists Y_1^n \exists A \varphi[x_1 \mapsto y_1, \ldots, x_p \mapsto y_p] \wedge \bigwedge_{i=1}^p \operatorname{EQ}(z_i, y_i) 
\] 
be the quantified Boolean formula obtained from $\Phi_1$ by first creating a copy $z_i$ of each of the universal variables $x_i$, replacing each universal variable $x_i$ of $\Phi_1$ with a new existential variable $y_i$, and then, for all $i\in \{1, 2, \ldots, p\}$, adding the enforcer $\operatorname{EQ}(z_i, y_i)$, where $z_i$ is a new universal variable. Furthermore, let $V_2' = Z_1^p \cup Y_1^n$, and let $V_2=V_2'\cup A$. By construction, each clause in $\Phi_2$ contains at most one universal variable and  each universal variable appears exactly once in~$\Phi_2$.

\begin{sublemma-prop}
$\Phi_1$ is a yes-instance of  {\sc $\forall \exists$ NAE-3-SAT}  if and only if $\Phi_2$ is a yes-instance of  {\sc $\forall \exists$ NAE-3-SAT}. 
\end{sublemma-prop}

\begin{proof}
First, suppose that $\Phi_1$ is a yes-instance of {\sc $\forall \exists$ NAE-3-SAT}. Let $\beta_1$ be a truth assignment for $V_1$ that nae-satisfies $\Phi_1$, and let $\beta_2'$ be the following truth assignment for $V_2'$:
\begin{enumerate}[label=(\roman*), noitemsep]
\item set $\beta_2'(y_i) = \beta_1(y_i)$ for each $i \in \{p+1,p+2,\ldots,n\}$;
\item set $\beta_2'(z_i) = \beta_1(x_i)$ for each $i \in \{1,2,\ldots,p\}$;
\item set $\beta_2'(y_i) = \beta_1(x_i)$ for each $i \in \{1,2,\ldots,p\}$.
\end{enumerate}
By (iii) and Observation~\ref{ob:EQ}, it follows that there is a truth assignment $\beta_2$ for $V_2$ that extends $\beta_2'$ such that, for each $i\in\{1,2,\ldots,p\}$, the clauses of $\operatorname{EQ}(z_i, y_i)$ are nae-satisfied. Furthermore, since $\Phi_1$ is nae-satisfiable for every truth assignment for $X_1^p$, if follows that $\Phi_2$ is nae-satisfiable for every truth assignment for $Z_1^p$. Hence, $\Phi_2$ is a yes instance of {\sc $\forall \exists$ NAE-3-SAT}.

Second, suppose that $\Phi_2$ is a yes-instance of {\sc $\forall \exists$ NAE-3-SAT}. Let $\beta_2$ be a truth assignment for $V_2$ that nae-satisfies $\Phi_2$. By Observation~\ref{ob:EQ}, it follows that $\beta_2(z_i) = \beta_2(y_i)$ for each $i \in \{1,2,\ldots,p\}$. Hence $\beta_1$ with $\beta_1(x_i) = \beta_2(z_i)$ for each $i \in \{1,2,\ldots,p\}$, and $\beta_1(y_i) = \beta_2(y_i)$ for each $i \in \{p+1,p+2,\ldots,n\}$ is a truth assignment for $V_1$ that nae-satisfies $\Phi_1$. Furthermore, since $\Phi_2$ is nae-satisfiable for every truth assignment for $Z_1^p$, it follows that $\Phi_1$ is a yes-instance of {\sc $\forall \exists$ NAE-3-SAT}.     
\end{proof}

Next, following Darmann and D\"ocker~\cite[Theorem 1]{darmann19}, we transform $\Phi_2$ into a new quantified Boolean formula in four steps: 
\begin{enumerate}
\item To remove all negated variables, we start by replacing each appearance of an existential variable in $Y_1^n$ with a new unnegated variable. Specifically, for each existential variable $y_i \in Y_1^n$, let $u(y_i)$ and $n(y_i)$ be the number of unnegated and negated appearances, respectively, of $y_i$ in the Boolean formula of $\Phi_2$. Recall that $u(y_i)+n(y_i)=a(y_i)$. Now, for each $j\in\{1,2,\ldots,u(y_i)\}$, replace the $j$-th unnegated appearance of $y_i$ in $\Phi_2$ with $y_{i,j}$. Similarly, for each $j\in\{1,2,\ldots,n(y_i)\}$, replace the $j$-th negated appearance of $y_i$ in $\Phi_2$ with $y_{i,u(y_i)+j}$. Lastly, for all $i \in \{1,2, \ldots, n\}$, append
\[
\exists y_{i,1} \exists y_{i,2} \cdots \exists y_{i,a(y_i)} 
\]   
to the list of existential variables and remove the obsolete variables~$\exists Y_1^n$.  
\item If $u(y_i)>1$, introduce the clauses
\[
\bigwedge_{i = 1}^n \quad\bigwedge_{j = 1}^{u(y_i)-1} \operatorname{EQ}(y_{i,j}, y_{i,j+1}).
\]  
Similarly, if $n(y_i)>1$, introduce the clauses
\[
\bigwedge_{i = 1}^n \quad \bigwedge_{j = u(y_i)+1}^{a(y_i)-1} \operatorname{EQ}(y_{i,j}, y_{i,j+1}). 
\]
\item For each $i\in\{1,2,\ldots,n\}$ with $u(y_i)\notin\{0,a(y_i)\}$, add the gadget $$\operatorname{NE}(y_{i, u(y_i)}, y_{i, u(y_i) + 1}).$$
\item Let $\Phi_2'$ be the quantified Boolean formula resulting from the last three steps. For $i\in \{1, 2, \ldots, n\}$, consider an existential variable $y_i$ in $\Phi_2$. If $y_i$  appears exactly once in $\Phi_2$, then $y_{i,1}$ only appears once in $\Phi_2'$ because Steps 2 and 3 do not introduce any gadget that adds an additional appearance of $y_{i,1}$. Otherwise, if  $y_i$ appears at  least twice in $\Phi_2$, then the enforcers introduced in the previous two steps increase the number of appearances for each variable $y_{i,j}$ with $j\in\{1,2,\ldots,a(y_i)\}$ by at least one and at most two. Hence, each variable $y_{i,j}$ appears at most three times in $\Phi_2'$. Moreover, by construction and Observation~\ref{ob:varInc}, each variable in~$A$ appears at most four times in $\Phi_2'$. Now, for each existential variable $v$ in $\Phi_2'$ (this includes all variables in $A$), add the clauses $$\bigwedge_{k=1}^{4-a(v)} \operatorname{P1}(v)$$ to $\Phi_2'$, where $a(v)$ denotes the number of appearances of $v$ in $\Phi_2'$.
\end{enumerate}
Let $\Phi_3$ be the quantified Boolean formula constructed by the preceding four-step procedure. Furthermore, let $V_3$ be the set of variables of $\Phi_3$, and let $V_3'=V_3-A$. 

\begin{sublemma-prop}
$\Phi_2$ is a yes-instance of {\sc $\forall \exists$ NAE-3-SAT}  if and only if $\Phi_3$ is a yes-instance of {\sc $\forall \exists$ NAE-3-SAT}. 
\end{sublemma-prop}

\begin{proof}
First, suppose that $\Phi_2$ is a yes-instance of {\sc $\forall \exists$ NAE-3-SAT}. Let $\beta_2$ be a truth assignment for $V_2$ that nae-satisfies $\Phi_2$. Obtain a truth assignment $\beta_3'$ for $V_3'$ as follows:
\begin{enumerate}[label=(\roman*), noitemsep]
\item set $\beta_3'(z_i) = \beta_2(z_i)$ for each $i \in \{1,2,\ldots,p\}$;
\item set $\beta_3'(y_{i,j}) = \beta_2(y_i)$ for each $i \in \{1,2,\ldots,n\}$ and $j \in \{1,2,\ldots,u(y_i)\}$;
\item set $\beta_3'(y_{i,j}) = \overline{\beta_2(y_i)}$ for each $i \in \{1,2,\ldots,n\}$ and $j \in \{u(y_i)+1,u(y_i)+2,\ldots,a(y_i)\}$.
\end{enumerate}
By (ii) and (iii) as well as Observations~\ref{ob:EQ},~\ref{ob:NE}, and~\ref{ob:padding-gadget}, it follows that there is a truth assignment $\beta_3$ for $V_3$ that extends $\beta_3'$ and nae-satisfies $\Phi_3$. Moreover, it follows by construction that for every truth assignment for $Z_1^p$, there exists a truth assignment for $V_3 - Z_1^p$ that nae-satisfies $\Phi_3$. Hence, $\Phi_3$ is a yes-instance of {\sc $\forall \exists$ NAE-3-SAT}.  

Second, suppose that $\Phi_3$ is a yes-instance of {\sc $\forall \exists$ NAE-3-SAT}. Let $\beta_3$ be a truth assignment for $V_3$ that nae-satisfies $\Phi_3$. By Steps 2 and 3 of the construction, and by Observations~\ref{ob:EQ} and \ref{ob:NE}, we have
\begin{enumerate}[label=(\Roman*), noitemsep]
\item $\beta_3(y_{i, 1}) = \beta_3(y_{i,2}) = \cdots = \beta_3(y_{i,u(y_i)})$,
\item $\beta_3(y_{i, u(y_i)}) \ne \beta_3(y_{i, u(y_i) + 1})$, and
\item $\beta_3(y_{i, u(y_i) + 1}) = \beta_3(y_{i,u(y_i) + 2}) = \cdots = \beta_3(y_{i,a(y_i)})$
\end{enumerate}
 for all $i \in \{1,2,\ldots,n\}$. Now, obtain a truth assignment $\beta_2$ for $V_2$ as follows:
\begin{enumerate}[label=(\roman*), noitemsep]
\item set $\beta_2(z_i) = \beta_3(z_i)$ for each $i \in \{1,2,\ldots,p\}$; 
\item set $\beta_2(y_i) = \beta_3(y_{i,1})$ for each $i \in \{1,2,\ldots,n\}$ with $u(y_i) \geq 1$;
\item set $\beta_2(y_i) = \overline{\beta_3(y_{i,1})}$ for each $i \in \{1,2,\ldots,n\}$ with $u(y_i) = 0$;
\item set $\beta_2(a) = \beta_3(a)$ for each $a\in A$ with $a\in V_2$.
\end{enumerate} 
It is now straightforward to check that $\beta_2$ nae-satisfies $\Phi_2$ and, hence, $\Phi_2$ is a yes-instance of {\sc $\forall \exists$ NAE-3-SAT}. 
\end{proof}

We complete the proof by showing that $\Phi_3$ has the desired properties. First, since all enforcers introduced in Section~\ref{sec:enforcer-2} are monotone, it follows from Step 1 in the construction of $\Phi_3$ from $\Phi_2$ that $\Phi_3$ is monotone. Second, again by Step 1 in the construction of $\Phi_3$ from $\Phi_2$, it follows that each clause in $\Phi_3$ is a 3-clause that contains three distinct variables. Third, turning to the universal variables in $\Phi_3$ and as mentioned in the construction of $\Phi_2$, each clause in $\Phi_2$, and hence in $\Phi_3$, contains at most one universal variable and each universal variable in $\Phi_2$, and hence in $\Phi_3$, appears exactly once. Fourth, recalling Step 4 in the construction of $\Phi_3$ from $\Phi_2$ and  that each new existential variable of $\operatorname{P1}(v)$ appears exactly four times in the seven clauses associated with $\operatorname{P1}(v)$, it follows that each existential variable appears exactly four times in $\Phi_3$. Noting that the size of $\Phi_3$ is polynomial in the size of $\Phi$, this establishes the proposition.
\end{proof}

We are now in a position to establish the main result of this section.

\begin{thm}\label{thm:nae}
{\sc Monotone $\forall \exists$ NAE-3-SAT-$(1, 3)$} is $\Pi_2^P$-complete if the Boolean formula is linear and each clause contains at most one universal variable. 
\end{thm}

\begin{proof}
Clearly, the decision problem {\sc Monotone $\forall \exists$ NAE-3-SAT-$(1, 3)$} as described in the statement of the theorem is in $\Pi_2^P$. We show $\Pi_2^P$-completeness by a reduction from {\sc Monotone \sc $\forall \exists$ NAE-3-SAT-$(1, 4)$}.  Let
\[
\Phi_1 = \forall Z_1^p \exists Z_{p+1}^n \varphi
\]
be an instance of {\sc Monotone $\forall \exists$ NAE-3-SAT-$(1, 4)$}. By Proposition~\ref{prop:nae}, we may assume that  each clause in $\Phi_1$ contains at most one universal variable.

We start by defining the following four sets of variables which we use to construct an instance $\Phi_2$ of {\sc Montone} $\forall\exists$ {\sc NAE-3-SAT-(1,3)}. Let 
\begin{align*}
&U=\{u_{i,k}: i\in\{p+1,p+2,\ldots,n\}\text{ and } k\in\{1,2,\ldots,8\}\}\text{ and}\\
&V=\{v_{i,k}: i\in\{p+1,p+2,\ldots,n\}\text{ and } k\in\{1,2,\ldots,8\}\}
\end{align*}
be two sets of universal variables, and let 
\begin{align*}
&E=\{e_{i,k}: i\in\{p+1,p+2,\ldots,n\}\text{ and } k\in\{1,2,\ldots,8\}\}\text{ and}\\
&Z=\{z_{i,k}: i\in\{p+1,p+2,\ldots,n\}\text{ and } k\in\{1,2,\ldots,8\}\}
\end{align*}
be two sets of existential variables. Now, for each $i \in \{p+1, p+2, \ldots, n\}$, we replace the $j$-th appearance of $z_i$ with $z_{i, j}$ for all $j \in \{1, 2, 3, 4\}$ and introduce the clauses
\begin{align*}
&\bigwedge_{k = 1}^7 \left ( (z_{i, k} \vee e_{i, k} \vee u_{i, k}) \wedge (e_{i, k} \vee z_{i, k + 1} \vee v_{i, k})\right ) \wedge\\
&(z_{i, 8} \vee e_{i, 8} \vee u_{i, 8}) \wedge (e_{i, 8} \vee z_{i, 1} \vee v_{i, 8}) \wedge {}\\
&(z_{i,5} \vee e_{i,1} \vee e_{i,2}) \wedge (z_{i,6} \vee e_{i,7} \vee e_{i,8}) \wedge (z_{i,7} \vee e_{i,3} \vee e_{i,4}) \wedge (z_{i,8} \vee e_{i,5} \vee e_{i,6}).
\end{align*}
Furthermore, we append each element in $U\cup V$ to the list of universal variables, append each element in $E\cup Z$ to the list of existential variables, and delete the obsolete variables $Z_{p+1}^n $. Let $\Phi_2$ denote the resulting formula. By construction, it is straightforward to check that $\Phi_2$ is an instance of {\sc Monotone $\forall \exists$ NAE-3-SAT-$(1, 3)$} with at most one universal variable per clause and whose set of variables is $$U\cup V\cup Z_1^p\cup E \cup Z.$$ Moreover, if any pair of clauses in $\Phi_1$ have two variables in common, then both are existential variables and, hence, again by construction, $\Phi_2$ is linear. Since $\Phi_2$ has all desired properties and the size of $\Phi_2$ is polynomial in the size of $\Phi_1$, it remains to show that the following statement holds.

\begin{sublemma}
$\Phi_1$ is a yes-instance of {\sc Monotone $\forall \exists$ NAE-3-SAT-$(1, 4)$} if and only if $\Phi_2$ is a yes-instance of {\sc Monotone $\forall \exists$ NAE-3-SAT-$(1, 3)$}. 
\end{sublemma}

\begin{proof}
First, suppose that $\Phi_1$ is a yes-instance of  {\sc Monotone $\forall \exists$ NAE-3-SAT-$(1, 4)$}. Let $\beta_1$ be a truth assignment for $Z_1^p\cup Z_{p+1}^n$ that nae-satisfies~$\Phi_1$. Obtain a truth assignment $\beta_2$ for $Z_1^p\cup E\cup Z$ as follows:

\begin{enumerate}[label=(\roman*), noitemsep]
\item set $\beta_2(z_i) = \beta_1(z_i)$ for each $i \in \{1,2,\ldots,p\}$;
\item set $\beta_2(z_{i,k}) = \beta_1(z_i)$ for each $i \in \{p+1,p+2,\ldots,n\}$ and $k\in\{1,2,\ldots,8\}$;
\item set $\beta_2(e_{i,k}) = \overline{\beta_1(z_i)}$ for each $i \in \{p+1,p+2,\ldots,n\}$ and $k\in\{1,2,\ldots,8\}$.
\end{enumerate}
It is easily checked that every truth assignment for $U\cup V\cup Z_1^p\cup E \cup Z$ that extends~$\beta_2$ nae-satisfies $\Phi_2$, and thus $\Phi_2$ is a yes-instance of {\sc Monotone $\forall \exists$ NAE-3-SAT-$(1, 3)$}.

Second, suppose that $\Phi_2$ is a yes-instance {\sc Monotone $\forall \exists$ NAE-3-SAT-$(1, 3)$}. Let $\beta_2$ be a truth assignment that nae-satisfies $\Phi_2$ such that $\beta_2(u_{i,k})=F$ and $\beta_2(v_{i,k})=T$ for each $i\in\{p+1,p+2,\ldots,n\}$ and $k\in\{1,2,\ldots,8\}$. Since $u_{i,k}$ and $v_{i,k}$ are universal variables, $\beta_2$ exists. We next show that $\beta_2$ satisfies the property  $$\beta_2(z_{i,1})=\beta_2(z_{i,2})=\beta_2(z_{i,3})=\beta_2(z_{i,4})$$ for each $i\in\{p+1,p+2,\ldots,n\}$. To this end, consider the subset of clauses
$$(z_{i, 8} \vee e_{i, 8} \vee F) \wedge (e_{i, 8} \vee z_{i, 1} \vee T)\wedge \bigwedge_{k = 1}^7 \left ( (z_{i, k} \vee e_{i, k} \vee F) \wedge (e_{i, k} \vee z_{i, k + 1} \vee T)\right )$$
of $\Phi_2$, where the universal variables are set according to $\beta_2$. If $\beta_2(z_{i,1})=F$, then the clause  $(z_{i, 1} \vee e_{i, 1} \vee F)$ implies that $\beta_2(e_{i, 1})=T$ and, hence, by the aforementioned subset of clauses, $\beta_2(z_{i,j})=F$ for each $j\in\{1,2,3,4\}$. Otherwise, if $\beta_2(z_{i,1})=T$, then the clause  $(e_{i, 8} \vee z_{i, 1} \vee T)$ implies that $\beta_2(e_{i, 8})=F$ and, hence, again by the aforementioned subset of clauses, $\beta_2(z_{i,j})=T$ for each $j\in\{1,2,3,4\}$. It now follows that the truth assignment $\beta_1$ for $Z_1^p\cup Z_{p+1}^n$ with
\begin{enumerate}[label=(\roman*), noitemsep]
\item $\beta_1(z_i) = \beta_2(z_i)$ for each $i \in \{1,2,\ldots,p\}$ and 
\item $\beta_1(z_{i}) = \beta_2(z_{i,1})$ for each $i \in \{p+1,p+2,\ldots,n\}$
\end{enumerate}
nae-satisfies $\Phi_1$, and so $\Phi_1$ is a yes-instance of  {\sc Monotone $\forall \exists$ NAE-3-SAT-$(1, 4)$}.
\end{proof}

This completes the proof of Theorem~\ref{thm:nae}.
\end{proof}

\subsection{Restrictions that alleviate the complexity of {\sc Monotone $\forall \exists$ NAE-3-SAT-$(s,t)$}}

In this section, we discuss variants of  {\sc Monotone $\forall \exists$ NAE-3-SAT-$(s,t)$} that are in co-NP or solvable in polynomial time. More precisely, we investigate the complexity of  {\sc Monotone $\forall \exists$ NAE-3-SAT-$(s,2)$}. First note that {\sc Monotone $\forall \exists$ NAE-3-SAT-$(0,t)$} is a special case of {\sc NAE-3-SAT} and therefore in NP. Furthermore, {\sc Monotone $\forall \exists$ NAE-3-SAT-$(s,1)$} can be solved in polynomial time since an instance of this problem is a yes-instance if and only if each clause contains at least one existential variable. Now consider the following decision problem that allows for a set of variables and the set $\{F,T\}$ of constants.

\noindent {\sc Monotone-with-Constants-NAE-3-SAT-$t$} ({\sc MC-NAE-3-SAT-$t$})\\
\noindent{\bf Input.} A positive integer $t$, a set $V=\{x_1,x_2,\ldots,x_n\}$ of variables and a monotone Boolean formula $$\bigwedge_{j=1}^m C_j$$ such that each clause contains exactly three distinct elements in $V\cup\{F,T\}$ and, amongst the clauses, each element in $V$
appears exactly $t$ times.  \\
\noindent{\bf Question.} Does there exist a truth assignment $\beta \colon V  \rightarrow \{T, F\}$ such that each clause of the formula is nae-satisfied?

\noindent Boolean formulas that include variables and the two constants $F$ and $T$ were, for example, previously considered in the context of NAE-3-SAT~\cite[p.\,275f]{bonet12}. In particular, let $\varphi$ be an instance of NAE-3-SAT that allows for constants. Bonet et al.\ ~\cite{bonet12} showed that, given a solution to $\varphi$, it is NP-complete to decide if a second solution to $\varphi$ exists. This result was in turn used to prove that two problems arising in computational biology are NP-complete. Note that in the special case in which $\varphi$ does not contain any constant, a second solution can always be obtained from a given  truth assignment that nae-satisfies $\varphi$ by simply interchanging $T$ and $F$. 

We next show that  {\sc MC-NAE-3-SAT-$t$} is solvable in polynomial time if $t=2$.

\begin{proposition}\label{thm:mcNAE3Sat2}
\textsc{MC-NAE-3-SAT-2} is in {\em P}. 
\end{proposition}

\begin{proof}
Let $\varphi=\bigwedge_{j=1}^m C_j$ be an instance of \textsc{MC-NAE-3-SAT-2} over a set $V\cup\{F,T\}$ of variables and constants, where $V=\{x_1,x_2,\ldots,x_n\}$. To establish the proposition, we adapt ideas presented by Porschen et al.~\cite{porschen04} who developed a linear-time algorithm to decide if an instance of \textsc{NAE-SAT}, i.e. a Boolean formula in CNF, is nae-satisfiable if each variable appears at most twice.

Using the notation $C_j=(\ell_{j,1}\vee\ell_{j,2}\vee\ell_{j,3})$ to denote the $j$-th clause in $\varphi$ for each $j\in\{1,2,\ldots,m\}$, we next present an algorithm to decide whether or not $\varphi$ is a yes-instance of \textsc{MC-NAE-3-SAT-2}. At each step of the algorithm, $\varphi$ is transformed into a simpler Boolean formula.
\begin{enumerate}
\item For each clause $C_j$ with  $\ell_{j,k}=x_i$ for some $i\in\{1,2,\ldots,n\}$ and  $\ell_{j,k'},\ell_{j,k''}\in\{F,T\}$ with $\{k,k',k''\}=\{1,2,3\}$, do the following.
\begin{enumerate}[label=(\Roman*), noitemsep]
\item If $\ell_{j,k'}\ne\ell_{j,k''}$, remove $C_j$ from $\varphi$.
\item If $\ell_{j,k'}=\ell_{j,k''}=F$, remove $C_j$ and reset $\varphi$ to be $\varphi[x_i \mapsto T]$.
\item If $\ell_{j,k'}=\ell_{j,k''}=T$, remove $C_j$ and reset $\varphi$ to be $\varphi[x_i \mapsto F]$.
\end{enumerate}
\item For each pair of variables $x_i, x_{i'} \in V$ with $i \neq i'$ that both appear in two distinct clauses $C_j$ and $C_{j'}$, remove $C_j$ and $C_{j'}$ from $\varphi$. 
\item For each variable $x_i$ that appears in exactly one clause $C_j$, remove $C_j$ from $\varphi$.   
\item For each clause $C_j$ such that $\ell_{j,1},\ell_{j,2}, \ell_{j,3} \in \{T, F\}$, do the following.
\begin{enumerate}[label=(\Roman*), noitemsep]
\item If $\{\ell_{j,1},\ell_{j,2}, \ell_{j,3}\}=\{T, F\}$, remove $C_j$ from $\varphi$.
\item Otherwise, stop and return ``$\varphi$ is a no-instance''.
\end{enumerate}
\item Stop and return ``$\varphi$ is a yes-instance''. 
\end{enumerate}  
The algorithm clearly terminates within polynomial time.  Moreover, as a variable that appears exactly once in a Boolean formula can be assigned to either $T$ or $F$ without affecting any other clause, it is  straightforward to check that each step in the algorithm returns a Boolean formula that is equivalent to $\varphi$. Hence, if a clause contains three equal constants, then the algorithm correctly returns that $\varphi$ is a no-instance in Step 4. Now, suppose that the algorithm returns ``$\varphi$ is a yes-instance''. Let $\varphi'$ be the Boolean formula that is obtained at the end of the last iteration of Step 4(I). Then, $\varphi'$ is an instance of \textsc{MC-NAE-3-SAT-2} such that each clause contains at most one constant and each pair of clauses have at most one variable in common. Hence, $\varphi'$ is linear.  It remains to show that $\varphi'$ is a yes-instance of \textsc{MC-NAE-3-SAT-2}.

Before continuing with the proof, we pause to give an overview of a result established by Porschen et al.~\cite{porschen04}. Let $\psi=\bigwedge_{j=1}^{m'} C'_{j}$ be a linear and monotone Boolean formula where each variable appears exactly twice and each clause contains at least two distinct variables. Furthermore, let $G_{\psi}$ be the {\it clause graph} for $\psi$ whose set of vertices is $\{C_1',C_2',\ldots,C'_{m'}\}$ and, for each pair $j,j'\in\{1,2,\ldots,m'\}$ with $j\ne j'$, there is an edge $\{C'_j, C'_{j'}\}$ in $G_{\psi}$ precisely if $C'_j$ and $C'_{j'}$ have a variable in common. Then $\psi$ is nae-satisfiable if and only if there exists an edge coloring of $G_{\psi}$ that uses exactly two colors $c_1$ and $c_2$ such that each vertex is incident to an edge that is colored $c_1$ and incident to an edge that is colored $c_2$. Moreover, if $\psi$ does not have a connected component that is isomorphic to a cycle of odd length, then such an edge coloring exists.

We now continue with the proof of the proposition. Let $\varphi''$ be the  Boolean formula obtained from $\varphi'$ by omitting all constants. By construction, each clause in $\varphi''$ contains either two or three distinct variables. It follows that, if $\varphi''$ is nae-satisfiable, then $\varphi'$ is nae-satisfiable. Let  $G_{\varphi''}$ be the clause graph for $\varphi''$. First, assume that $G_{\varphi''}$ does not have a connected component that is isomorphic to a cycle of odd length. Then it immediately follows from the result by Porschen et al.~\cite{porschen04} that $\varphi''$ is nae-satisfiable and, hence, $\varphi'$ is also nae-satisfiable. Second, assume that $G_{\varphi''}$ has a connected component that is isomorphic to a cycle of odd length. Then the vertices of this component are of the form
\[
(x_{i_1}\vee x_{i_2}), (x_{i_2}\vee x_{i_3}), \ldots, (x_{i_{p-1}}\vee x_{i_p}), (x_{i_p}\vee x_{i_1}),
\]
where $p\ge 3$ is an odd integer and $x_{i_j}\in V$. In other words,
\[
\mathcal{C_{\varphi''}}= \bigwedge_{j = 1}^{p-1} (x_{i_j} \vee x_{i_{j+1}})\wedge (x_{i_p} \vee x_{i_1}),  
\]
is contained in $\varphi''$. Although $\mathcal{C_{\varphi''}}$ is not nae-satisfiable, we next show that the corresponding clauses $\mathcal{C_{\varphi'}}$ in $\varphi'$ are nae-satisfiable since each such clause contains exactly one constant. 

Consider
\[
\mathcal{C_{\varphi'}} =  \bigwedge_{j = 1}^{p-1} (x_{i_j} \vee x_{i_{j+1}} \vee b_j)\wedge (x_{i_p} \vee x_{i_1} \vee b_p),  
\]
where $b_j \in \{T, F\}$  for each $j\in\{1,2,\ldots,p\}$. Let $\beta$ be the following truth assignment for $\{x_{i_1},x_{i_2},\ldots,x_{i_p}\}$:
\begin{enumerate}[label=(\roman*), noitemsep]
\item set $\beta(x_{i_j}) = \overline{b_p}$ for each $j \in \{1,2,\ldots,p\}$ with $j$ being odd;
\item set $\beta(x_{i_j}) = b_p$ for each $j \in \{1,2,\ldots,p\}$ with $j$ being even.
\end{enumerate}
It follows that $\beta$ nae-satisfies $\mathcal{C_{\varphi'}}$. An analogous argument can be applied to every other connected component in $G_{\varphi''}$ that is isomorphic to a cycle of odd length. Furthermore, it again follows from Porschen et al.'s result~\cite{porschen04} that the edge set of each connected component in $G_{\varphi''}$ that is not isomorphic to a cycle of odd length corresponds to a subset of clauses in $\varphi''$ that is nae-satisfiable. Altogether, $\varphi''$ is nae-satisfiable and, hence, $\varphi'$ is also nae-satisfiable. This completes the proof of the proposition.
\end{proof}

We next establish three corollaries that pinpoint the complexity of  {\sc Monotone $\forall \exists$ NAE-3-SAT-$(s,2)$}.

\begin{cor}
{\sc Monotone $\forall \exists$ NAE-3-SAT-$(s, 2)$} is in {\em co-NP} for any fixed positive integer $s$. 
\end{cor}
\begin{proof}
A no-instance of {\sc Monotone $\forall \exists$ NAE-3-SAT-$(s, 2)$} can be identified by taking an assignment of the universal variables and applying the algorithm presented in Proposition~\ref{thm:mcNAE3Sat2} to verify in polynomial time whether or not the resulting \textsc{MC-NAE-3-SAT-2} Boolean formula (with omitted lists of universal and existential quantifies) is not nae-satisfiable.
\end{proof}

\begin{cor}\label{cor:NAE-trivial}
{\sc Monotone $\forall \exists$ NAE-3-SAT-$(s, 2)$} is trivially a yes-instance for any fixed positive integer $s$ if each clause contains at most one universal variable. 
\end{cor}
\begin{proof}
Let $\Phi$ be an instance of {\sc Monotone $\forall \exists$ NAE-3-SAT-$(s, 2)$} such that each clause contains at most one universal variable. We follow ideas that are similar to those presented in the algorithm described in the proof of Proposition~\ref{thm:mcNAE3Sat2}. First, if there are two existential variables that both appear in two distinct clauses  $C_j$ and $C_{j'}$, obtain a new Boolean formula by removing $C_j$ and $C_{j'}$ from $\Phi$. Repeat this step until no such pair of variables remains. Then, if there is an existential variable that appears in exactly one clause $C_j$,  obtain a new Boolean formula by removing $C_j$. Similar to the proof of Proposition~\ref{thm:mcNAE3Sat2}, it follows that the resulting Boolean formula, say $\Phi'$, is an instance of {\sc Monotone $\forall \exists$ NAE-3-SAT-$(s, 2)$} such that each clause contains at most one universal variable and the formula is linear. Moreover, $\Phi$ is a yes-instance if and only if $\Phi'$ is a yes-instance. If $\Phi'$ is empty, then $\Phi$ is a yes-instance by correctness of the applied transformations. Otherwise, it follows from the properties of~$\Phi'$ and the proof of Proposition~\ref{thm:mcNAE3Sat2} that $\Phi'$ and, hence, $\Phi$ are yes-instances.     
\end{proof}

\begin{cor}
{\sc Monotone $\forall \exists$ NAE-3-SAT-$(1, 2)$} is in {\em P}.
\end{cor}

\begin{proof}
Let $\Phi$ be an instance of {\sc Monotone $\forall \exists$ NAE-3-SAT-$(1, 2)$}. To decide whether or not $\Phi$ is a yes-instance, we apply the following algorithm to $\Phi$.
\begin{enumerate}
\item If there exists a clause that contains three distinct universal variables, then stop and return ``$\Phi$ is a no-instance''.
\item For a clause $C_j$  that contains one existential variable, say $x$, and two distinct universal variables, say $u$ and $u'$, let $C_{j'}$ be the unique clause that contains the second appearance of $x$. Then remove $C_j$ and turn $x$ in $C_{j'}$  into a new universal variable. 
\begin{enumerate}[label=(\Roman*), noitemsep]
\item If $C_{j'}$ now contains three distinct universal variables, stop and return ``$\Phi$ is a no-instance''.
\item Otherwise, repeat until there are no clauses with two universal variables.
\end{enumerate}
\item Stop and return ``$\Phi$ is a yes-instance''.
\end{enumerate}
Since each universal variable appears exactly once in $\Phi$ and each  existential variable appears exactly twice in~$\Phi$, it follows that the Boolean formula obtained after each iteration of Step 2 is an instance of {\sc Monotone $\forall \exists$ NAE-3-SAT-$(1, 2)$}. Therefore, if the algorithm eventually produces a Boolean formula, $\Phi'$ say, then $\Phi’$ has at most one universal variable in each clause and, by Corollary~\ref{cor:NAE-trivial}, $\Phi'$ is a yes-instance. Hence, to see that the algorithm works correctly, it suffices to show that an iteration of Step~2 preserves yes-instances. Suppose that $\Phi_1$ is the quantified Boolean formula at the start of an iteration of Step~2 and $C_j=(x\vee u\vee u')$ is a clause in $\Phi_1$ as described in Step~2. Let $\beta$ be a truth assignment that nae-satisfies $\Phi$. If $\beta(u) = \beta(u') = F$, then it follows that $\beta(x) = T$. On the other hand, if $\beta(u) = \beta(u') = T$, then this implies that $\beta(x) = F$. It now follows that the Boolean formula, $\Phi_2$ say, obtained by turning $x$ into a universal variable is also a yes-instance. Conversely, by reversing this argument, if $\Phi_2$ is a yes-instance, then $\Phi_1$ is a yes-instance. Thus Step~2 preserves yes-instances. 
We now establish the corollary by noting that the described algorithm has a running time that is polynomial in the size of $\Phi$. 
\end{proof}

\noindent{\bf Acknowledgements.} The third and fourth authors thank the New Zealand Marsden Fund for their financial support.

\end{document}